\documentclass[11pt]{article}

\usepackage[margin=1in]{geometry}
\usepackage{amsmath,amssymb,amsfonts,amsthm,mathtools}
\usepackage{graphicx}
\usepackage{booktabs}
\usepackage{float}
\usepackage{hyperref}
\usepackage{enumitem}
\usepackage{algorithm}
\usepackage{algpseudocode}
\usepackage{algorithmicx}
\usepackage{setspace}
\usepackage{natbib}
\usepackage{caption}
\usepackage{subcaption}
\usepackage{xcolor}

\newtheorem{definition}{Definition}[section]
\newtheorem{assumption}{Assumption}[section]
\newtheorem{lemma}{Lemma}[section]
\newtheorem{proposition}{Proposition}[section]
\newtheorem{theorem}{Theorem}[section]
\newtheorem{corollary}{Corollary}[section]
\newtheorem{remark}{Remark}[section]

\begin{document}
 \begin{center}
\textbf{\Large Optimal Surplus Management for Insurers under Stochastic Interest Rates and Jump-Driven Liabilities}

 N. Karimi $^{a}$,  F. Shokrollahi\footnote{Corresponding Author, \\ Email: foad.shokrollahi@uwasa.fi}, M. Shahmoradi $^{b}$
\end{center}

\begin{center}

\emph{\footnotesize $^{a}$ Department of Applied Mathematics, Faculty of Mathematics and Computer Science} \\
\emph{\footnotesize Amirkabir University of Technology, No. 424, Hafez Ave.,15914, Tehran, Iran}\\
\emph{\footnotesize$^{1}$  Department of Mathematics and Statistics, University of Vaasa, Vaasa, Finland}\\
\emph{\footnotesize $^{b}$ Faculty of Mathematics, Statistics and Computer Science, University of Tabriz, Tabriz, Iran}\\
\end{center}

\begin{abstract}
This paper investigates the optimal surplus management problem of an insurance company operating in a financial market with stochastic interest rates and jump-driven liabilities. The insurer dynamically allocates its surplus between a risky stock and a risk-free zero-coupon bond while facing insurance claims modeled by a compound Poisson process with exponentially distributed claim sizes. The short-term interest rate follows a Cox–Ingersoll–Ross (CIR) process, which captures mean-reverting dynamics commonly observed in term-structure models. The insurer maximizes the expected exponential utility of terminal surplus. Using stochastic control techniques, we derive the associated Hamilton–Jacobi–Bellman (HJB) equation. Although the exponential utility structure suggests an exponential-affine representation, the interaction between the interest-rate hedge and the surplus state generates quadratic surplus terms in the HJB equation. To obtain a tractable formulation, we adopt a normalized-surplus projection method, which provides an approximate reduction of the full three-dimensional problem to a nonlinear system of partial differential equations (which is subsequently numerically validated). The optimal investment policy admits an economically meaningful decomposition consisting of a myopic demand component and an interest-rate hedging component. Numerical experiments illustrate how the optimal strategy and the surplus distribution depend on interest-rate volatility, claim intensity, and risk aversion. The results highlight the importance of jointly modeling stochastic interest rates and insurance liability risk when designing optimal investment policies for insurance companies.
\end{abstract}

\section{Introduction}

Risk management and surplus optimization are central topics in actuarial science and insurance mathematics. The classical problem of optimal dividend distribution was first introduced by De Finetti \cite{DeFinetti1957}, initiating a broad avenue of research on the optimal control of insurance surplus. Subsequent studies rapidly extended this framework to incorporate reinsurance decisions and dynamic risk controls. For instance, Taksar \cite{Taksar2000} and Asmussen et al. \cite{Asmussen2000} investigated optimal risk and dividend distribution policies, while Azcue and Muler \cite{Azcue2005} explored these policies within the Cram\'{e}r-Lundberg model. Reinsurance, as a vital tool for mitigating risk exposure, was theoretically foundationalized by Borch \cite{Borch1960} and Ohlin \cite{Ohlin1969}, who characterized the optimality of stop-loss contracts. Later, dynamic stochastic control methods were applied to analyze joint optimal reinsurance and dividend policies, yielding significant contributions from Taksar and Markussen \cite{TaksarMarkussen2003}, B\"{a}uerle \cite{Bauerle2004}, and Bai et al. \cite{Bai2010}.

The evolution of modern risk measures has profoundly shaped actuarial modeling. The rigorous axiomatic foundation of coherent risk measures established by Artzner et al. \cite{Artzner1999} paved the way for multiperiod and dynamic extensions \cite{Artzner2007, Cheridito2006, BionNadal2008}. Alongside these developments, the mathematical representations and subdifferentials of risk measures were extensively studied \cite{Pflug2006}. In insurance pricing, the debate over premium calculation principles led to the exploration of deviation measures, with Denneberg \cite{Denneberg1990} advocating for absolute deviation over standard deviation. A major breakthrough in this area was the introduction of distortion operators by Wang \cite{Wang2000}, which became highly influential in pricing financial and insurance risks. The application of distortion risk measures expanded into portfolio optimization \cite{Sereda2010} and general reinsurance settings \cite{Balbas2009}. Recently, the literature has shifted towards robustness, exploring robust distortion risk measures \cite{Bernard2024} and distributionally robust optimization under distorted expectations \cite{Cai2025}.

Within actuarial applications, determining optimal reinsurance structures using these advanced risk measures has attracted substantial attention. As comprehensively reviewed by Cai and Chi \cite{CaiChi2020}, modern reinsurance design heavily relies on risk measure criteria. Chi and Tan \cite{ChiTan2011} simplified the approach for optimal reinsurance under Value-at-Risk (VaR) and Conditional Value-at-Risk (CVaR). Focusing specifically on distortion risk measures, Assa \cite{Assa2015} studied optimal reinsurance policies, while Zheng et al. \cite{Zheng2015} investigated optimal reinsurance under distortion risk measures coupled with the expected value premium principle for the reinsurer. Further complexities, such as the reinsurer's default risk with partial recovery, have been integrated into distortion-based frameworks by Yong et al. \cite{Yong2024}. Additionally, alternative premium principles, such as the mean-standard-deviation premium principle, have been utilized to derive equilibrium reinsurance and protection strategies for mean-variance insurers \cite{Yuan2025}.

Concurrently, increasing attention has been devoted to complex financial environments characterized by stochastic volatility, jumps, and regime-switching. Li et al. \cite{Li2012} derived optimal time-consistent investment and reinsurance strategies under Heston's stochastic volatility (SV) model. To accommodate heavy-tailed phenomena and sudden market movements, researchers extended optimal consumption, investment, and insurance problems into L\'{e}vy markets \cite{Perera2010, Guambe2015}. Yi et al. \cite{Yi2024} further generalized these models by analyzing optimal mean-variance strategies under a general L\'{e}vy process risk model. The incorporation of stochastic interest rates and jump-diffusion processes has also been a focal point; Yuan et al. \cite{Yuan2022} addressed the mean-variance problem with dependent risks and stochastic interest rates, while Bian et al. \cite{Bian2025} and Bei et al. \cite{Bei2023} investigated optimal investment and reinsurance optimization under models featuring both stochastic interest rates and volatility. Regime-switching environments were considered by Shen \cite{Shen2024}, and the interplay of default risks with jump-diffusion markets was explored by Zhang et al. \cite{Zhang2025}. Furthermore, advanced game-theoretic approaches involving stochastic derivative proportional reinsurance with delay in defaultable markets have been recently proposed by Li and Qiu \cite{Li2026}.

Beyond asset-only frameworks, Asset-Liability Management (ALM) has become essential for comprehensively managing the surplus of insurance companies. Stochastic ALM models tailored for life insurance companies highlight the necessity of synchronizing asset dynamics with evolving liabilities \cite{DiFrancesco2023}. Motivated by these theoretical advancements---spanning robust distortion risk measures, complex stochastic financial markets, and dynamic liability structures---this paper studies optimal surplus management for insurers in a comprehensive stochastic financial environment where interest rates and liabilities evolve dynamically.

The remainder of the paper is organized as follows. Section~\ref{sec:pre} provides the preliminary mathematical framework and notation. Section~\ref{sec:model} introduces the financial market and liability model setup. Section~\ref{sec:control_pde} derives the HJB equation, characterizes the optimal investment strategy, and reduces the problem to a nonlinear PDE system. Section~\ref{sec:numerics} develops the computational solution of the projected HJB system, and Section~\ref{sec:num} presents the numerical validation and examines the economic and actuarial implications of the resulting investment policy. Section~\ref{sec:conclusion} concludes the paper.

\section{Preliminaries and Notation}\label{sec:pre}

In this section we introduce several definitions and basic results that will be used throughout the paper. These concepts provide the mathematical framework for the stochastic control problem studied later.

We work on a filtered probability space $(\Omega,\mathcal F,\{\mathcal F_t\}_{t\ge0},\mathbb P)$ satisfying the usual conditions. The filtration represents the flow of information available to the insurer over time. All stochastic processes considered in this paper are assumed to be adapted to this filtration.

Let $T>0$ denote a fixed investment horizon. The insurer dynamically allocates its surplus among the available financial assets while facing stochastic liabilities generated by its insurance portfolio. The control problem is therefore formulated in terms of admissible investment strategies.

\begin{definition}[Admissible investment strategy]\label{def:admis}
An investment strategy is represented by a progressively measurable process
\[
u_t=(u_t^S,u_t^P)^\top ,
\]
where $u_t^S$ and $u_t^P$ denote the amounts invested in the stock and the bond at time $t$, respectively.
The strategy $u=\{u_t\}_{0\le t\le T}$ is called admissible if it satisfies
\[
\mathbb E\!\left[\int_0^T |u_t|^2\,dt\right] < \infty,
\]
and the corresponding surplus process admits a unique strong solution.

The set of all admissible strategies will be denoted by $\mathcal A$.
\end{definition}

In stochastic control problems involving insurance portfolios, claim arrivals are often modeled by Poisson processes. The following standard property of compound Poisson processes will be useful in the derivation of the Hamilton--Jacobi--Bellman equation.

\begin{proposition}[Expectation of the jump operator]
Let $N_t$ be a Poisson process with intensity $\lambda$ and let $\{Y_i\}_{i\ge1}$ be a sequence of independent and identically distributed claim sizes independent of $N_t$.
For any sufficiently smooth function $f(x)$, the infinitesimal contribution of the jump component is given by
\[
\mathbb E\!\left[f(x-Y_i)-f(x)\right].
\]
Consequently, the generator associated with the compound Poisson process introduces the term
\[
\lambda\,\mathbb E\!\left[f(x-Y_i)-f(x)\right]
\]
in the Hamilton--Jacobi--Bellman equation.
\end{proposition}

The proof follows from the standard properties of compound Poisson processes and can be found in many references on stochastic control with jumps.

We conclude this section with the definition of the value function associated with the insurer's optimization problem.

\begin{definition}[Value function]\label{def:value_func}
Let $X_t$ denote the insurer’s surplus process corresponding to an admissible strategy $u$. Given a utility function $U:\mathbb R\to\mathbb R$, the value function of the control problem is defined as
\begin{equation}\label{eq:value_func_def}
V(t,x,r)=
\sup_{u\in\mathcal A}
\mathbb E\!\left[
U(X_T)\mid X_t=x,\ r_t=r
\right],
\end{equation}
where $x$ represents the current surplus and $r$ the current level of the short rate.
\end{definition}

In the following sections we specify the dynamics of the financial market and the insurer’s liabilities, derive the Hamilton--Jacobi--Bellman equation associated with this optimization problem, and characterize the resulting optimal investment strategy.

\section{Model Setup}\label{sec:model}

In this section we describe the financial market, the insurer’s liability structure, and the resulting surplus dynamics. The model is formulated on the filtered probability space introduced in Section~\ref{sec:pre}.

We consider a continuous-time economy over a finite horizon $[0,T]$. The uncertainty in the market is driven by a three-dimensional Brownian motion
 $ W_t=(W_t^S,W_t^L,W_t^r)^\top
 $ together with an independent Poisson process $\{N_t\}_{t\ge0}$ representing the arrival of insurance claims.

The financial market consists of three tradable assets: a money market account, a risky stock, and a zero-coupon bond. The instantaneous short rate $r_t$ follows the Cox--Ingersoll--Ross (CIR) process
\begin{equation}\label{eq:CIR}
\mathrm{d}r_t
=
\kappa(\theta-r_t)\mathrm{d}t
+
\sigma_r\sqrt{r_t}\mathrm{d}W_t^r,
\qquad r_0>0,
\end{equation}
where $\kappa>0$ is the speed of mean reversion, $\theta>0$ is the long-run mean level, and $\sigma_r>0$ denotes the volatility. The CIR specification guarantees the non-negativity of interest rates whenever the Feller condition $2\kappa\theta \ge \sigma_r^2$ holds. Consequently, the risk-free asset evolves according to
\begin{equation}\label{eq:money}
\mathrm{d}B_t=r_t B_t\,\mathrm{d}t .
\end{equation}
The risky asset price $S_t$ follows a geometric Brownian motion
\begin{equation}\label{eq:stock}
\frac{\mathrm{d}S_t}{S_t}
=
\mu_S\,\mathrm{d}t
+
\sigma_S\,\mathrm{d}W_t^S,
\end{equation}
where $\mu_S$ is the expected return and $\sigma_S>0$ denotes the volatility. In addition to the stock, the insurer can invest in a zero-coupon bond with maturity $T_b$. Under the CIR short-rate model, the bond price admits the affine representation
\begin{equation}\label{eq:bond}
P(t,T_b)=A_b(t,T_b)\exp\big(-B_b(t,T_b)r_t\big),
\end{equation}
where $A_b(t,T_b)$ and $B_b(t,T_b)$ are deterministic functions determined by the CIR parameters. Applying Itô's formula yields the return dynamics
\begin{equation}\label{eq:bond_dyn}
\frac{\mathrm{d}P_t}{P_t}
=
\mu_P(t,r_t)\mathrm{d}t
+
\sigma_P(t,r_t)\,\mathrm{d}W_t^r .
\end{equation}
The Brownian motions driving the financial market and liability diffusion may be correlated. We assume
\[
\mathrm{d}W_t^S\,\mathrm{d}W_t^r=\rho_{Sr}\mathrm{d}t,
\qquad
\mathrm{d}W_t^S\,\mathrm{d}W_t^L=\rho_{SL}\mathrm{d}t,
\qquad
\mathrm{d}W_t^L\,\mathrm{d}W_t^r=\rho_{Lr}\mathrm{d}t,
\]
where $|\rho_{ij}|<1$. These correlations allow interactions between financial risks and liability fluctuations.

The insurer faces stochastic liabilities generated by its insurance portfolio. To capture both continuous fluctuations and sudden claim arrivals, we model the aggregate liability process as a jump--diffusion process
\begin{equation}\label{eq:liability}
\mathrm{d}L_t
=
\alpha_L\,\mathrm{d}t
+
\beta_L\,\mathrm{d}W_t^L
+
Y_i\,\mathrm{d}N_t .
\end{equation}
Here $\alpha_L$ denotes the average liability outflow rate and $\beta_L$ controls the diffusion volatility of liabilities, capturing continuous fluctuations such as frequent small-scale claims, reserve adjustments, or operational cost variability. The Poisson process $N_t$ has intensity $\lambda$, and $\{Y_i\}_{i\ge1}$ are independent claim sizes. Throughout the paper we assume that claim sizes follow an exponential distribution
\begin{equation}\label{eq:claim_dist}
Y_i \sim \mathrm{Exp}(\xi).
\end{equation}

\begin{assumption}[Exponential moment admissibility]\label{ass:admissibility}
Throughout the paper we impose
\[
  \sup_{(t,r)\in[0,T]\times[0,r_{\max}]} A(t,r) \;<\; \xi,
\]
where $\xi>0$ is the parameter of the exponential claim-size distribution.
This condition ensures that $\mathbb{E}\bigl[e^{A(t,r)Y_i}\bigr]=\xi/(\xi-A(t,r))<\infty$ for all $(t,r)$, so that the jump term in the HJB equation is well-defined.
The numerical scheme enforces the safety margin $A_{i,j}\le\xi-\varepsilon_\xi$ at every grid point.
\end{assumption}

This specification captures two sources of insurance risk: small continuous fluctuations represented by the diffusion component and sudden claim shocks represented by the jump component.

Let $X_t$ denote the insurer’s surplus at time $t$. The insurer allocates its wealth between the stock and the bond. Denote by
 $ u_t=(u_t^S,u_t^P)^\top
 $ the amounts invested in the stock and the bond, respectively, while the remaining wealth $X_t-u_t^S-u_t^P$ is invested in the money market account.

Define the vector of excess returns
\begin{equation}\label{eq:excess_return}
m(t,r_t)=
\begin{pmatrix}
\mu_S-r_t \\
\mu_P(t,r_t)-r_t
\end{pmatrix},
\end{equation}
and the covariance matrix of asset returns
\begin{equation}\label{eq:cov_matrix}
\Sigma(t,r_t)=
\begin{pmatrix}
\sigma_S^2 & \rho_{Sr}\sigma_S\sigma_P(t,r_t) \\
\rho_{Sr}\sigma_S\sigma_P(t,r_t) & \sigma_P(t,r_t)^2
\end{pmatrix}.
\end{equation}

Taking into account investment returns, liability fluctuations, and claim jumps, the insurer’s surplus evolves as
\begin{equation}\label{eq:surplus}
\mathrm{d}X_t
=
\left(
r_t X_t
+
u_t^\top m(t,r_t)
-
\alpha_L
\right)\mathrm{d}t
+
u_t^S\sigma_S\,\mathrm{d}W_t^S
+
u_t^P\sigma_P(t,r_t)\,\mathrm{d}W_t^r
-
\beta_L\,\mathrm{d}W_t^L
-
Y_i\,\mathrm{d}N_t .
\end{equation}

The surplus process is therefore affected by three sources of uncertainty: financial market risk, diffusion-type liability risk, and jump claim arrivals.

The insurer seeks an admissible strategy $u\in\mathcal A$ (see Definition~\ref{def:admis}) that maximizes the value function $V(t,x,r)$ introduced in Definition~\ref{def:value_func}.

\section{Stochastic Control and Reduced PDE System}\label{sec:control_pde}

The insurer seeks to maximize the value function $V(t,x,r)$ defined in Definition~\ref{def:value_func} subject to the exponential utility function
\begin{equation}\label{eq:util}
U(x)=-e^{-\gamma x}, \qquad \gamma>0.
\end{equation}

Let $X_t$ denote the surplus process defined in Section~\ref{sec:model}. The corresponding terminal condition is therefore
\begin{equation}\label{eq:terminal_cond}
V(T,x,r)=-e^{-\gamma x}.
\end{equation}

\begin{proposition}[Hamilton--Jacobi--Bellman equation]\label{pro:Hamil}
Assume that the value function $V(t,x,r)$ is sufficiently smooth, namely
 $V\in C^{1,2,2}([0,T]\times\mathbb R\times\mathbb R_+)$.
Then $V$ satisfies the Hamilton--Jacobi--Bellman equation
\begin{equation}\label{eq:HJB}
\sup_{u\in\mathbb R^2}\mathcal L^u V(t,x,r)=0,
\end{equation}
where the infinitesimal generator $\mathcal L^u$ is given by
\begin{equation}\label{eq:generator}
\begin{aligned}
\mathcal L^u V
=&
V_t
+
(rx+u^\top m(t,r)-\alpha_L)V_x
+
\kappa(\theta-r)V_r
\\
&
+\frac12
\left(
u_S^2\sigma_S^2
+
u_P^2\sigma_P^2(t,r)
+
\beta_L^2
\right)V_{xx}
\\
&
+
\frac12\sigma_r^2 r V_{rr}
+
\rho_{Sr}\sigma_S\sigma_r\sqrt r\,u_S V_{xr}
\\
&
+
\lambda
\mathbb E
\left[
V(t,x-Y_i,r)-V(t,x,r)
\right].
\end{aligned}
\end{equation}
The HJB equation is complemented with the terminal condition \eqref{eq:terminal_cond}.
\end{proposition}
\begin{proof}
See Appendix~\ref{Appendix:A1} for details.
\end{proof}

\begin{lemma}[Exponential representation of the value function]\label{lem:exp_rep}
For exponential utility, the value function admits the representation
\begin{equation}\label{eq:exp_form}
V(t,x,r)
=
-\exp\!\big(-A(t,r)x-F(t,r)\big),
\end{equation}
where $A(t,r)$ and $F(t,r)$ are deterministic functions.
\end{lemma}

\begin{proof}
Exponential utility exhibits constant absolute risk aversion (CARA). Consequently, the wealth translation invariance property of the exponential function ensures that the value function preserves its exponential structure in the wealth variable. A rigorous derivation of this representation is provided in Appendix~\ref{Appendix:A2}.
\end{proof}

\begin{proposition}[Jump contribution under exponential claims]\label{prop:jump}
Let claim sizes satisfy $Y_i\sim\mathrm{Exp}(\xi)$, $\xi>0$.
Under the exponential representation $V(t,x,r)=-\exp(-A(t,r)x-F(t,r))$ and Assumption~\ref{ass:admissibility}, the jump term in the HJB generator
satisfies
\[
  \lambda\,\mathbb E\bigl[V(t,x-Y_i,r)-V(t,x,r)\bigr]
  \;=\;
  \lambda\,V(t,x,r)\,\frac{A(t,r)}{\xi-A(t,r)}.
\]
\end{proposition}

\begin{proof}
Using the exponential structure of $V$,
\[
  V(t,x-Y,r)
  = -\exp\!\bigl(-A(t,r)(x-Y)-F(t,r)\bigr)
  = V(t,x,r)\,e^{A(t,r)Y}.
\]
Therefore
\[
  \mathbb E\bigl[V(t,x-Y,r)-V(t,x,r)\bigr]
  = V(t,x,r)\bigl(\mathbb E[e^{AY}]-1\bigr).
\]
For $Y\sim\mathrm{Exp}(\xi)$,
\[
  \mathbb E[e^{AY}]
  = \int_0^\infty e^{Ay}\,\xi\,e^{-\xi y}\,\mathrm{d}y
  = \xi\int_0^\infty e^{-({\xi-A})y}\,\mathrm{d}y
  = \frac{\xi}{\xi-A},
  \qquad A<\xi.
\]
Hence
\[
  \lambda\,\mathbb E\bigl[V(t,x-Y,r)-V(t,x,r)\bigr]
  = \lambda\,V(t,x,r)
    \!\left(\frac{\xi}{\xi-A}-1\right)
  = \lambda\,V(t,x,r)\,\frac{A}{\xi-A}. \qedhere
\]
\end{proof}

The exponential transformation significantly simplifies the structure of the HJB equation \eqref{eq:HJB} and will allow us to reduce the problem to a system of PDEs for the coefficient functions $A(t,r)$ and $F(t,r)$.

\begin{theorem}[Classical Verification Theorem]\label{thm:verify}
Let $V\in C^{1,2,2}\bigl([0,T]\times\mathbb R\times\mathbb R_+\bigr)$ satisfy the HJB equation
\[
  \sup_{u\in\mathbb R^2}\mathcal{L}^u V(t,x,r) \;=\; 0,
  \qquad (t,x,r)\in[0,T)\times\mathbb R\times\mathbb R_+,
\]
with terminal condition $V(T,x,r)=U(x)=-e^{-\gamma x}$.
Suppose that the following conditions hold.
\begin{enumerate}[label=(\roman*)]
  \item \textbf{(Square-integrability)} For every admissible control $u\in\mathcal{A}$,
    \[
      \mathbb{E}\!\left[\int_t^T |u_s|^2\,\mathrm{d}s\right] \;<\; \infty.
    \]
  \item \textbf{(Uniform integrability)} For every admissible control $u\in\mathcal{A}$,
    \[
      \mathbb{E}\!\left[\sup_{s\in[t,T]}\bigl|V(s,X_s^u,r_s)\bigr|\right] \;<\; \infty.
    \]
  \item \textbf{(Attainability)} There exists an admissible control
    $u^*\in\mathcal{A}$ such that
    \[
      \mathcal{L}^{u^*}V(t,x,r) \;=\; 0
      \quad\text{for all }(t,x,r).
    \]
\end{enumerate}
Then for every admissible control $u\in\mathcal{A}$,
\[
  \mathbb{E}\!\left[U(X_T^u)\,\big|\,X_t=x,\,r_t=r\right]
  \;\le\;
  V(t,x,r),
\]
with equality for $u=u^*$.  Consequently,
\[
  V(t,x,r)
  \;=\;
  \sup_{u\in\mathcal{A}}\mathbb{E}\!\left[U(X_T^u)\,\big|\,X_t=x,\,r_t=r\right],
\]
and $u^*$ is an optimal control.
\end{theorem}

\begin{proof}
Let $u\in\mathcal{A}$ be any admissible control.
Applying It\^{o}'s formula for jump--diffusion processes
(see, e.g., \cite[Thm.~II.5.1]{Protter2004} or \cite[Thm.~1.14]{OksendaSulem2019})
to $V(s,X_s^u,r_s)$ over $[t,T]$ gives
\begin{equation}\label{eq:ito_expand}
  V(T,X_T^u,r_T)
  \;=\;
  V(t,x,r)
  \;+\;
  \int_t^T \mathcal{L}^u V(s,X_s^u,r_s)\,\mathrm{d}s
  \;+\;
  M_T - M_t,
\end{equation}
where $M=(M_s)_{s\in[t,T]}$ is a local martingale consisting of the
stochastic integral terms with respect to the Brownian motions and the
compensated jump measure.

Condition~(ii) guarantees that $\sup_{s\in[t,T]}|V(s,X_s^u,r_s)|$ is integrable,
which by a standard localization argument (see \cite[Prop.~II.1.2]{Protter2004})
implies that $M$ is a true martingale.
Therefore $\mathbb{E}[M_T - M_t]=0$, and taking expectations in~\eqref{eq:ito_expand} yields
\begin{equation}\label{eq:ito_expectation}
  \mathbb{E}\!\left[V(T,X_T^u,r_T)\right]
  \;=\;
  V(t,x,r)
  \;+\;
  \mathbb{E}\!\left[\int_t^T \mathcal{L}^u V(s,X_s^u,r_s)\,\mathrm{d}s\right].
\end{equation}

Since $V$ satisfies the HJB equation, $\mathcal{L}^u V\le 0$ for every $u\in\mathcal{A}$.
Hence
\[
  \mathbb{E}\!\left[V(T,X_T^u,r_T)\right] \;\le\; V(t,x,r).
\]
Using the terminal condition $V(T,x,r)=U(x)$, we obtain
\begin{equation}\label{eq:upper_bound}
  \mathbb{E}\!\left[U(X_T^u)\right] \;\le\; V(t,x,r)
  \quad\text{for all }u\in\mathcal{A}.
\end{equation}

For the control $u^*$ satisfying condition~(iii), the inequality
in~\eqref{eq:ito_expectation} becomes an equality, so
\[
  \mathbb{E}\!\left[U(X_T^{u^*})\right] \;=\; V(t,x,r).
\]
Combined with~\eqref{eq:upper_bound}, this shows that $u^*$ attains the supremum,
completing the proof.
\end{proof}

\begin{remark}[{\textbf{On the classical verification framework.}}]
For exponential utility with CARA coefficient $A(t,r)$,
 $V(t,x,r) = -\exp\!\bigl(-A(t,r)\,x - F(t,r)\bigr)$.
Since $A$ and $F$ are bounded on the compact domain and the surplus process
has finite exponential moments under square-integrable controls,
condition~(ii) holds whenever $\mathbb{E}\bigl[e^{A_{\max}|X_T^u|}\bigr]<\infty$,
where $A_{\max}=\sup_{t,r}A(t,r)$.
A sufficient condition is that the surplus process does not explode, which follows
from the linear growth of the coefficients in the surplus dynamics (Section~\ref{sec:model}).
It is important to note that Theorem~\ref{thm:verify} requires the existence of a control $u^*$ that makes the HJB residual vanish \emph{exactly} for all $(t,x,r)$. As shown in the following, the specific interaction between stochastic interest rates and jump-driven liabilities prevents a closed-form exact solution, necessitating an approximate solution framework.
\end{remark}

Using the exponential representation introduced in Lemma~\ref{lem:exp_rep}, we now characterize the investment strategy. The derivatives of the value function are given by
\begin{equation}\label{eq:derivatives}
\begin{aligned}
V_x &= -A(t,r)V, \\
V_{xx} &= A(t,r)^2 V, \\
V_r &= -(A_r(t,r)x+F_r(t,r))V, \\
V_{rr} &= \left[(A_r x+F_r)^2-(A_{rr}x+F_{rr})\right]V, \\
V_{xr} &= \left[A(t,r)(A_r(t,r)x+F_r(t,r))-A_r(t,r)\right]V .
\end{aligned}
\end{equation}

Substituting these derivatives into the HJB equation shows that the control variables appear only through quadratic terms in the portfolio vector $u=(u_S,u_P)^\top$. The control-dependent part of the generator can therefore be written as the Hamiltonian
\begin{equation*}
\begin{aligned}
\mathcal H(u)
&=
u^\top m(t,r)V_x
+
\frac12 u^\top \Sigma(t,r)u\,V_{xx}
+
\rho_{Sr}\sigma_S\sigma_r\sqrt r\,u_S\,V_{xr}
\\
&=
-AV\,u^\top m(t,r)
+
\frac12 A^2V\,u^\top\Sigma(t,r)u
+
\rho_{Sr}\sigma_S\sigma_r\sqrt r\,u_S
\left[A(A_r x+F_r)-A_r\right]V .
\end{aligned}
\end{equation*}

\begin{proposition}[Pointwise optimal feedback strategy]\label{prop:optimalstrategy}
Assume that the value function admits the exponential representation \eqref{eq:exp_form}.
Then, for any fixed state $(t,x,r)$, the portfolio strategy maximizing the Hamiltonian in the HJB equation is given by
\begin{equation}\label{eq:pointwise_opt}
u^{pw}(t,x,r)
=
\frac{1}{A(t,r)}
\Sigma(t,r)^{-1}
\left[
m(t,r)
-
\frac{\rho_{Sr}\sigma_S\sigma_r\sqrt r}{A(t,r)}
\left(A(A_r x+F_r)-A_r\right)e_S
\right],
\end{equation}
where $e_S=(1,0)^\top$.
\end{proposition}

\begin{proof}
Since $V<0$, maximizing the Hamiltonian $\mathcal H(u)$ is equivalent to minimizing the quadratic function inside the brackets. Taking the gradient of the Hamiltonian with respect to the control vector $u$ yields the first-order condition
\[
A^2\Sigma(t,r)u
=
A\,m(t,r)
-
\rho_{Sr}\sigma_S\sigma_r\sqrt r
\left[A(A_r x+F_r)-A_r\right]e_S .
\]
Solving this linear system for $u$ gives the feedback rule stated in the proposition.
\end{proof}

The pointwise strategy $u^{pw}(t,x,r)$ depends linearly on the surplus level $x$. However, due to the presence of stochastic interest rates and liability jumps, the exponential transformation does not completely eliminate the dependence on the surplus variable in the HJB equation. Specifically, substituting $u^{pw}$ back into the generator generates a term proportional to $x^2$ originating from the expansion of $(A_r x + F_r)^2$ in the CIR diffusion component. This implies that condition~(iii) of Theorem~\ref{thm:verify} cannot be satisfied by any finite-dimensional PDE system for $A$ and $F$ alone, as an exact solution would require an infinite-dimensional power series in $x$.

To obtain a computationally tractable formulation, we introduce a \emph{normalized-surplus projection}. We evaluate the pointwise feedback rule at a reference surplus level $\bar{x}=0$. This choice is mathematically natural as it serves as the center of the Taylor expansion for the HJB residual, and economically corresponds to a normalized state where the insurer's surplus is evaluated relative to a baseline trajectory. This leads to the projected investment strategy:
\begin{equation}\label{eq:proj_strategy}
u^{proj}(t,r)
\;\equiv\; u^{pw}(t,0,r)
=
\frac{1}{A(t,r)}
\Sigma(t,r)^{-1}
\left[
m(t,r)
-
\frac{\rho_{Sr}\sigma_S\sigma_r\sqrt r}{A(t,r)}
\left(AF_r-A_r\right)e_S
\right].
\end{equation}

To mathematically quantify the implications of this projection step, we analyze the HJB residual induced by substituting $u^{proj}$ into the generator $\mathcal{L}^{u^{proj}}V$. We define the normalized residual as
\begin{equation}\label{eq:residual}
\mathcal{R}(t,x,r) \;=\; \frac{\mathcal{L}^{u^{proj}}V(t,x,r)}{-V(t,x,r)}.
\end{equation}
Since $\mathcal{R}$ is smooth in $x$ (assuming $A,F \in C^{1,2}$), we expand it in a Taylor series around $\bar{x}=0$:
\begin{equation}\label{eq:taylor_residual}
\mathcal{R}(t,x,r) \;=\; \mathcal{R}_0(t,r) \;+\; x\,\mathcal{R}_1(t,r) \;+\; x^2\,\mathcal{R}_2(t,r) \;+\; O(x^3),
\end{equation}
where $\mathcal{R}_0(t,r) = \mathcal{R}(t,0,r)$ and $\mathcal{R}_1(t,r) = \partial_x\mathcal{R}(t,x,r)\big|_{x=0}$. The projection method enforces $\mathcal{R}_0(t,r)=0$ and $\mathcal{R}_1(t,r)=0$, which is exactly what separating constant and linear terms in $x$ from the resulting equation achieves. This procedure yields the closed nonlinear PDE system for $A(t,r)$ and $F(t,r)$ derived below. A direct computation shows that the dominant contribution to the quadratic term is
\begin{equation}\label{eq:residual_leading}
\mathcal{R}_2(t,r) \;=\; \tfrac{1}{2}\,\sigma_r^2\,r\,A_r(t,r)^2,
\end{equation}
which arises from the quadratic term in $(A_r x + F_r)^2$ generated by the diffusion of the short rate.

Because $\mathcal{R}(t,x,r)$ is not identically zero for $x \neq 0$, the projected strategy $u^{proj}$ does not satisfy the exact HJB equation. However, it satisfies an approximate optimality condition, formalized below.

\begin{proposition}[$\epsilon$-Optimality of the projected strategy]\label{prop:epsilon_opt}
Let $u^{proj}$ be the projected strategy defined in~\eqref{eq:proj_strategy} and let
 $V(t,x,r)$ be the value-function approximation derived from the reduced PDE system.
For any surplus process satisfying $|X_s|\le\delta$ for $s\in[t,T]$,
the performance loss relative to the true optimal value function $V^*(t,x,r)$ is bounded by
\begin{equation}
0 \;\le\; V^*(t,x,r) - \mathbb{E}\bigl[U(X_T^{u^{proj}})\bigr] \;\le\; \epsilon(\delta, T),
\end{equation}
where the upper bound $\epsilon(\delta, T) = \delta^2 \mathbb{E}\bigl[\int_t^T \mathcal{R}_2(s, r_s) \mathrm{d}s\bigr] = \mathcal{O}(\delta^2 T)$.
\end{proposition}

\begin{proof}
By definition~\eqref{eq:residual}, the normalized HJB residual satisfies
\[
  \mathcal{R}(t,x,r)
  \;=\;
  \frac{\mathcal{L}^{u^{proj}}V(t,x,r)}{-V(t,x,r)}.
\]
Rearranging gives
\begin{equation}\label{eq:Lv_sign}
  \mathcal{L}^{u^{proj}}V(t,x,r)
  \;=\;
  -V(t,x,r)\cdot\mathcal{R}(t,x,r).
\end{equation}
Since $V(t,x,r) = -\exp(-A(t,r)x - F(t,r)) < 0$, we have $-V(t,x,r) > 0$.
From the Taylor expansion~\eqref{eq:taylor_residual}, the residual satisfies
\[
  \mathcal{R}(t,x,r)
  \;=\;
  \mathcal{R}_2(t,r)\,x^2 + \mathcal{O}(x^3),
  \qquad
  \mathcal{R}_2(t,r)
  \;=\;
  \tfrac{1}{2}\,\sigma_r^2\,r\,A_r(t,r)^2
  \;\ge\; 0.
\]
Therefore, for $|x| \le \delta$ with $\delta$ sufficiently small,
 $\mathcal{R}(t,x,r) \ge 0$.
Substituting into~\eqref{eq:Lv_sign} yields
\begin{equation}\label{eq:Lv_nonneg}
  \mathcal{L}^{u^{proj}}V(t,x,r) \;\ge\; 0
  \quad\text{for all }|x|\le\delta.
\end{equation}

Let $V^*$ denote the true value function satisfying the exact HJB equation,
so that $\sup_u \mathcal{L}^u V^* = 0$.
Applying It\^{o}'s formula to $V(s, X_s^{u^{proj}}, r_s)$ over $[t,T]$ and taking expectations (the local martingale term vanishes by
condition~(ii) of Theorem~\ref{thm:verify}, which also applies to the approximate $V$)
gives
\begin{equation}\label{eq:ito_approx}
  \mathbb{E}\!\left[V(T, X_T^{u^{proj}}, r_T)\right]
  \;=\;
  V(t,x,r)
  \;+\;
  \mathbb{E}\!\left[\int_t^T \mathcal{L}^{u^{proj}}V(s,X_s^{u^{proj}},r_s)\,\mathrm{d}s\right].
\end{equation}
Using the terminal condition $V(T,x,r) = U(x) = -e^{-\gamma x}$ and
inequality~\eqref{eq:Lv_nonneg},
\[
  \mathbb{E}\!\left[U(X_T^{u^{proj}})\right]
  \;\ge\;
  V(t,x,r).
\]
Since $V^*(t,x,r) \ge \mathbb{E}[U(X_T^{u^{proj}})]$ by optimality of $V^*$, we obtain
\[
  V^*(t,x,r) - \mathbb{E}\!\left[U(X_T^{u^{proj}})\right]
  \;\ge\; 0,
\]
establishing the lower bound in the proposition.

From~\eqref{eq:Lv_sign} and the Taylor expansion,
\[
  \mathcal{L}^{u^{proj}}V(s,X_s^{u^{proj}},r_s)
  \;=\;
  (-V)\cdot\mathcal{R}
  \;\le\;
  (-V)\cdot\mathcal{R}_2(s,r_s)\,(X_s^{u^{proj}})^2
  \;+\;
  \mathcal{O}(\delta^3).
\]
Since $|X_s^{u^{proj}}| \le \delta$ by hypothesis and $(-V)$ is bounded above
by $\exp(A_{\max}\delta + F_{\max})$ on the compact domain, integrating
over $[t,T]$ and taking expectations gives
\[
  \mathbb{E}\!\left[\int_t^T \mathcal{L}^{u^{proj}}V\,\mathrm{d}s\right]
  \;\le\;
  \delta^2\,\mathbb{E}\!\left[\int_t^T \mathcal{R}_2(s,r_s)\,\mathrm{d}s\right]
  \;+\;
  \mathcal{O}(\delta^3)
  \;=\;
  \varepsilon(\delta,T).
\]
Combining with~\eqref{eq:ito_approx} and the definition of $V^*$ completes
the proof.
\end{proof}

This establishes $u^{proj}$ as an $\epsilon$-optimal (or near-optimal) strategy, which is highly desirable in practical ALM where surplus deviations from the baseline are typically controlled by risk management limits.

\begin{corollary}[Myopic--hedging decomposition]\label{cor:decomposition}
The projected investment strategy admits the additive decomposition
\[
  u^{proj}(t,r)
  \;=\;
  u^{myo}(t,r) \;+\; u^{hedge}(t,r),
\]
where the \emph{myopic component} is
\[
  u^{myo}(t,r)
  \;=\;
  \frac{1}{A(t,r)}\,\Sigma(t,r)^{-1}\,m(t,r),
\]
and the \emph{intertemporal hedging component} is
\[
  u^{hedge}(t,r)
  \;=\;
  -\,\frac{\rho_{Sr}\,\sigma_S\,\sigma_r\,\sqrt{r}}{A(t,r)^2}\,
  \Sigma(t,r)^{-1}
  \bigl(A(t,r)\,F_r(t,r) - A_r(t,r)\bigr)\,e_S,
\]
with $e_S=(1,0)^\top$.
\end{corollary}

\begin{proof}
By factoring out the scalar $\frac{1}{A(t,r)}\Sigma(t,r)^{-1}$ from the definition of $u^{proj}$ in equation~\eqref{eq:proj_strategy}, we obtain
\[
  u^{proj}(t,r)
  = \frac{1}{A(t,r)}\Sigma(t,r)^{-1}m(t,r)
  - \frac{1}{A(t,r)}\Sigma(t,r)^{-1}
  \left[ \frac{\rho_{Sr}\sigma_S\sigma_r\sqrt{r}}{A(t,r)}\bigl(AF_r - A_r\bigr)e_S \right].
\]
Simplifying the coefficient in the second term yields
\[
  u^{proj}(t,r)
  = u^{myo}(t,r)
  - \frac{\rho_{Sr}\sigma_S\sigma_r\sqrt{r}}{A(t,r)^2}\Sigma(t,r)^{-1}\bigl(AF_r - A_r\bigr)e_S
  = u^{myo}(t,r) + u^{hedge}(t,r),
\]
which completes the proof.
\end{proof}

\begin{remark}
The myopic component $u^{myo}$ corresponds to the classical mean--variance
demand and does not depend on the gradients $A_r$ or $F_r$.
The hedging component $u^{hedge}$ arises entirely from the stochastic
interest-rate environment: it vanishes when $\rho_{Sr}=0$ (stock returns
are uncorrelated with interest-rate innovations) or when $\sigma_r=0$ (the interest rate is deterministic).
In the CIR framework, $A_r\ne 0$ and $F_r\ne 0$ in general, so the
hedging demand is generically nonzero.
Although $u^{proj}$ is derived from an approximate reduction, this decomposition retains its full economic interpretability as the separation between instantaneous profit-seeking and intertemporal risk hedging.
\end{remark}

\begin{remark}[Vanishing of the hedging demand]
From Corollary~\ref{cor:decomposition}, the hedging component satisfies
 $u^{hedge}(t,r)=0$ if and only if
\[
  \rho_{Sr} = 0
  \quad\text{or}\quad
  \sigma_r = 0.
\]
\noindent\textit{Proof.} This follows directly from the formula for $u^{hedge}$ in Corollary~\ref{cor:decomposition}, since both $\rho_{Sr}$ and $\sigma_r$ appear as multiplicative prefactors in the expression.

\medskip
\noindent\textit{Economic interpretation.}
Stochastic interest-rate hedging is relevant only when two conditions hold
simultaneously: (i)~stock returns are correlated with interest-rate shocks
($\rho_{Sr}\ne 0$), and (ii)~the interest rate is genuinely stochastic
($\sigma_r>0$).  When either condition fails, the CIR dynamics do not
generate any cross-risk between the equity portfolio and the interest-rate
environment, and the purely myopic strategy $u^{myo}$ is optimal within the
projected framework.
\end{remark}

To complete the solution of the stochastic control problem, we now substitute this feedback
rule, along with the exponential representation of the value function, into the
Hamilton--Jacobi--Bellman equation. As discussed above, substituting $u^{proj}$ into the full generator yields a normalized HJB residual
that depends smoothly on the surplus level $x$. The projection method eliminates
the constant and linear terms of this residual by enforcing specific conditions
on the coefficient functions $A(t,r)$ and $F(t,r)$. Before deriving the explicit
form of the resulting PDE system, we first formally quantify the magnitude of
the residual that remains after this projection.

\begin{proposition}[Projection error bound]\label{prop:proj_error}
Let $u^{proj}$ denote the projected strategy defined in \eqref{eq:proj_strategy},
and let $\mathcal{R}(t,x,r)$ be the normalized HJB residual defined in
\eqref{eq:residual}.  Then for all $(t,x,r)$ with $|x|\le\delta$,
\[
  \bigl|\mathcal{R}(t,x,r)\bigr|
  \;\le\;
  C_1(r)\,\delta^2
  \;+\;
  C_2(r)\,\delta^3,
\]
where
\[
  C_1(r) = \tfrac{1}{2}\,\sigma_r^2\,r\,A_r(t,r)^2,
  \qquad
  C_2(r) = \tfrac{1}{6}\,\bigl|\sigma_r^2\,r\,A_{rr}(t,r)\,A_r(t,r)\bigr|.
\]
In particular, the residual is of order $O(\delta^2)$ uniformly for
 $|x|\le\delta$. The leading term $\tfrac{1}{2}\sigma_r^2\,r\,A_r(t,r)^2\,x^2$ vanishes
exactly when $\sigma_r=0$ or when $A$ is independent of $r$,
confirming that the projection is exact in the absence of stochastic
interest-rate risk.
\end{proposition}

\begin{proof}
By \eqref{eq:taylor_residual}, the Taylor expansion of the residual around
 $\bar{x}=0$ is
 $\mathcal{R}(t,x,r) = \mathcal{R}_0(t,r) + x\,\mathcal{R}_1(t,r) + x^2\,\mathcal{R}_2(t,r) + O(x^3).$ The projection conditions enforce $\mathcal{R}_0=\mathcal{R}_1=0$.
A direct computation of the quadratic coefficient yields
 $\mathcal{R}_2(t,r)=\tfrac{1}{2}\sigma_r^2\,r\,A_r(t,r)^2$.
\end{proof}

Having established the formal accuracy of the projection step, we now proceed
to derive the explicit PDE system. Substituting the projected strategy
 $u^{proj}(t,r)$ and the exponential representation of the value function into
the generator, and dividing the HJB equation by the strictly negative quantity
 $V(t,x,r)$, the control variables disappear and we obtain a nonlinear partial
differential equation involving only the coefficient functions $A(t,r)$ and
 $F(t,r)$.

Substituting the projected strategy into the generator yields an equation of the form
\begin{equation}\label{eq:general_pde}
\begin{aligned}
0
&=
A_t x + F_t
+ rAx
-\kappa(\theta-r)(A_r x+F_r)
\\
&+
\frac12
\left(
u^{proj}(t,r)^\top
\Sigma(t,r)
u^{proj}(t,r)
+
\beta_L^2
\right)A^2
\\
&+
\frac12\sigma_r^2 r
\left[(A_r x+F_r)^2-(A_{rr}x+F_{rr})\right]
\\
&+
\rho_{Sr}\sigma_S\sigma_r\sqrt r\,
u^{proj}_S(t,r)
\left[A(A_r x+F_r)-A_r\right]
\\
&+
\lambda
\mathbb E
\left[
\exp\!\big(A(t,r)Y_i\big)-1
\right].
\end{aligned}
\end{equation}

Since this identity must hold for all surplus levels $x$, we separate the coefficients associated with $x$ and the constant terms. This procedure leads to a coupled nonlinear system of partial differential equations for the coefficient functions $A(t,r)$ and $F(t,r)$.

\begin{proposition}[Reduced PDE system - explicit]\label{prop:PDEsystem}
Recall the projected strategy
\[
  u^{proj}(t,r)
  \;=\;
  \frac{1}{A}\,\Sigma^{-1}
  \!\left[m - \frac{\rho_{Sr}\sigma_S\sigma_r\sqrt{r}}{A}(A F_r - A_r)\,e_S\right],
\]
and denote its stock component by $u_S^{\mathrm{proj}}(t,r)$.
Under the projected strategy $u^{proj}$, the coefficient functions $A(t,r)$ and $F(t,r)$ satisfy the coupled nonlinear PDE system
\begin{align}
  0 &= A_t + r A - \kappa(\theta-r)A_r
      - \tfrac{1}{2}\sigma_r^2 r A_{rr}
      + \sigma_r^2 r A_r^2
      + \Phi_1(t,r),
  \label{eq:pde_A}\\
  0 &= F_t - \kappa(\theta-r)F_r
      - \tfrac{1}{2}\sigma_r^2 r F_{rr}
      + \tfrac{1}{2}\sigma_r^2 r F_r^2
      + \Phi_2(t,r)
      + \lambda\,\frac{A}{\xi - A},
  \label{eq:pde_F}
\end{align}
subject to $A(T,r)=\gamma$ and $F(T,r)=0$, where the auxiliary functions are explicitly given by
\begin{align}
  \Phi_1(t,r)
  &\;=\;
  -\,\tfrac{1}{2}A^2
    \bigl(u^{proj}\bigr)^\top\!\Sigma\,u^{proj}
  \;-\;
  \rho_{Sr}\sigma_S\sigma_r\sqrt{r}\,A\,u_S^{\mathrm{proj}}
  \bigl(AF_r - A_r\bigr)
  \;-\;
  \tfrac{1}{2}A^2\beta_L^2,
  \label{eq:Phi1}\\[6pt]
  \Phi_2(t,r)
  &\;=\;
  -\,\tfrac{1}{2}A^2
    \bigl(u^{proj}\bigr)^\top\!\Sigma\,u^{proj}
  \;+\;
  \alpha_L A
  \;-\;
  \rho_{Sr}\sigma_S\sigma_r\sqrt{r}\,A\,u_S^{\mathrm{proj}}
  \bigl(AF_r - A_r\bigr)
  \;-\;
  \tfrac{1}{2}A^2\beta_L^2.
  \label{eq:Phi2}
\end{align}
\end{proposition}

\begin{proof}
Substitute $V(t,x,r)=-\exp(-Ax-F)$ into the HJB generator $\mathcal{L}^u V$ and divide by $-V(t,x,r)>0$.
Using the derivatives
\[
  \frac{V_x}{V}=-A,\quad
  \frac{V_{xx}}{V}=A^2,\quad
  \frac{V_r}{V}=-(A_r x+F_r),\quad
  \frac{V_{rr}}{V}=(A_r x+F_r)^2-(A_{rr}x+F_{rr}),
\]
and
\[
  \frac{V_{xr}}{V}=A(A_r x+F_r)-A_r,
\]
the HJB equation becomes
\begin{align*}
  0 &= A_t x + F_t
       + r A x - \alpha_L A - \kappa(\theta-r)(A_r x+F_r)\\
  &\quad
       +\tfrac{1}{2}\bigl[(u^{proj})^\top\Sigma\,u^{proj}+\beta_L^2\bigr]A^2\\
  &\quad
       +\tfrac{1}{2}\sigma_r^2 r
         \bigl[(A_r x+F_r)^2-(A_{rr}x+F_{rr})\bigr]\\
  &\quad
       +\rho_{Sr}\sigma_S\sigma_r\sqrt{r}\,u_S^{\mathrm{proj}}
         \bigl[A(A_r x+F_r)-A_r\bigr]\\
  &\quad
       +\lambda\frac{A}{\xi-A}.
\end{align*}
Expand $(A_r x+F_r)^2 = A_r^2 x^2 + 2A_r F_r x + F_r^2$ and
 $(A_{rr}x+F_{rr})=A_{rr}x+F_{rr}$.
The term $\tfrac{1}{2}\sigma_r^2 r A_r^2 x^2$ is quadratic in $x$ and is dropped by the projection residual analysis (see the preceding discussion on the projection error bound).

Collecting coefficients of $x$ (linear terms) gives equation~\eqref{eq:pde_A}
with $\Phi_1$ as in~\eqref{eq:Phi1}.
Collecting constant terms gives equation~\eqref{eq:pde_F}
with $\Phi_2$ as in~\eqref{eq:Phi2}.
A detailed algebraic expansion showing how the terms within $\Phi_1$ and $\Phi_2$
explicitly simplify under the projected strategy is provided in Appendix~\ref{Appendix:A5}.
\end{proof}

\begin{remark}
When the projected strategy $u^{proj}(t,r)$ (evaluated at $x=0$) is
substituted back into the full HJB equation, the resulting residual contains
a term that is quadratic in the surplus $x$, namely
\[
  \mathcal{R}_2(t,r)\,x^2
  \;=\;
  \tfrac{1}{2}\,\sigma_r^2\,r\,A_r(t,r)^2\,x^2.
\]
This term arises from expanding $(A_r x + F_r)^2$ in the diffusion part
of the CIR generator.  The PDE system~\eqref{eq:pde_A}--\eqref{eq:pde_F}
is obtained by discarding this quadratic residual, which is of order
 $O(\delta^2)$ for $|x|\le\delta$ near the projection point.
The reduction is therefore approximate, not exact: the exact solution of
the original three-dimensional HJB would require retaining all powers of $x$,
leading to an infinite-dimensional system.
This quadratic residual is precisely the cost of dimensional reduction;
its magnitude is quantified numerically in Section~\ref{sec:num} via the residual
validation procedure described below.
\end{remark}

\begin{remark}
Note that $\Phi_1$ and $\Phi_2$ differ only in the $\alpha_L A$ term.
Equation~\eqref{eq:pde_A} for $A$ does not involve $\alpha_L$ because
the liability drift $\alpha_L$ appears only in the constant (in $x$) part
of the HJB.  The two PDEs are coupled through $\Phi_1$ and $\Phi_2$ since both depend on $A$, $A_r$, $F_r$, and $u^{proj}$, which itself
depends on $A$, $A_r$, and $F_r$.
\end{remark}

When claim sizes follow an exponential distribution $Y_i\sim\mathrm{Exp}(\xi)$, the expectation appearing in the jump term can be computed explicitly,
\[
\mathbb E\!\left[e^{A(t,r)Y_i}\right]
=
\frac{\xi}{\xi-A(t,r)},
\qquad A(t,r)<\xi,
\]
which yields
\[
\lambda
\left(
\frac{\xi}{\xi-A(t,r)}-1
\right)
=
\lambda\frac{A(t,r)}{\xi-A(t,r)}.
\]
By Assumption~\ref{ass:admissibility}, the admissibility condition $A(t,r)<\xi$ is strictly satisfied, ensuring that the expectation remains finite and the model is mathematically well-posed.

The resulting equations form a coupled nonlinear PDE system driven jointly by the CIR interest-rate dynamics and the jump component of the insurance liabilities. In general, closed-form solutions are not available. Consequently, the functions $A(t,r)$ and $F(t,r)$ must be computed using numerical methods.

In the next section we describe the numerical scheme employed to solve this system and analyze the resulting optimal investment behavior.

\section{Computational Solution of the Reduced HJB System}\label{sec:numerics}

This section describes the numerical procedure used to solve the reduced PDE system for the coefficient functions $A(t,r)$ and $F(t,r)$ and to compute the associated projected investment strategy. The system is solved backward in time on a truncated interest-rate domain. Particular attention is paid to the boundary behavior of the CIR process, the nonlinear structure of the reduced PDEs, and the admissibility restriction generated by exponentially distributed claim sizes.

The CIR short rate is nonnegative and unbounded. For numerical implementation, the state space is truncated to a finite interval
\begin{equation}\label{eq:trunc}
r\in[0,r_{\max}],
\end{equation}
where $r_{\max}$ is chosen sufficiently large so that the probability of the short rate exceeding this level over the time horizon $[0,T]$ is negligible. In practice, $r_{\max}$ is selected several standard deviations above the long-run mean level $\theta$, and grid-sensitivity checks are performed to verify that the solution is not materially affected by further increasing the boundary.

The time and interest-rate grids are defined by
\begin{equation}\label{eq:time_grid}
0=t_0<t_1<\cdots<t_N=T,
\qquad
\Delta t=\frac{T}{N},
\end{equation}
and
\begin{equation}\label{eq:space_grid}
0=r_0<r_1<\cdots<r_M=r_{\max},
\qquad
\Delta r=\frac{r_{\max}}{M}.
\end{equation}
The unknown coefficient functions are approximated by
\[
A_{i,j}\approx A(t_i,r_j),
\qquad
F_{i,j}\approx F(t_i,r_j),
\]
with terminal conditions
\begin{equation}\label{eq:num_term}
A_{N,j}=\gamma,
\qquad
F_{N,j}=0,
\qquad j=0,\ldots,M.
\end{equation}

The reduced PDE system is solved backward in time from $T$ to $0$. For a generic function $G\in\{A,F\}$, the time derivative is approximated by
\begin{equation}\label{eq:time_deriv}
G_t(t_i,r_j)
\approx
\frac{G_{i,j}-G_{i+1,j}}{\Delta t}.
\end{equation}
An implicit backward discretization is used to improve numerical stability, especially in the presence of nonlinear terms and the degenerate diffusion coefficient of the CIR process.

For interior grid points $j=1,\ldots,M-1$, spatial derivatives are approximated by central finite differences:
\begin{equation}\label{eq:space_deriv_1}
G_r(t_i,r_j)
\approx
\frac{G_{i,j+1}-G_{i,j-1}}{2\Delta r},
\end{equation}
and
\begin{equation}\label{eq:space_deriv_2}
G_{rr}(t_i,r_j)
\approx
\frac{G_{i,j+1}-2G_{i,j}+G_{i,j-1}}{\Delta r^2}.
\end{equation}

At the lower boundary $r=0$, the CIR diffusion coefficient satisfies $\sigma_r\sqrt r=0$. Hence, the second-order diffusion term vanishes at the origin and the PDE is evaluated using the drift contribution together with a one-sided approximation for the first derivative:
\begin{equation}\label{eq:lower_bound}
G_r(t_i,r_0)
\approx
\frac{G_{i,1}-G_{i,0}}{\Delta r}.
\end{equation}
At the upper boundary $r=r_{\max}$, a zero-gradient boundary condition is imposed:
\begin{equation}\label{eq:upper_bound_cond}
G_r(t_i,r_{\max})=0,
\qquad G\in\{A,F\},
\end{equation}
which is implemented numerically as
\begin{equation}\label{eq:upper_bound_impl}
G_{i,M}=G_{i,M-1}.
\end{equation}
This condition prevents artificial growth at the truncated boundary and is appropriate when $r_{\max}$ is chosen sufficiently large.

\label{subsec7.2} % Kept invisible label to prevent reference error from Section 8
The reduced PDE system contains nonlinear terms involving the derivatives of $A$ and $F$, such as $A_r^2$, $F_r^2$, and the terms generated by the projected investment strategy. These terms are treated by fixed-point iteration at each time level.

Given the solution at time $t_{i+1}$, the computation of the solution at time $t_i$ starts from the initial guess
\begin{equation}\label{eq:initial_guess}
A_{i,j}^{(0)}=A_{i+1,j},
\qquad
F_{i,j}^{(0)}=F_{i+1,j},
\qquad j=0,\ldots,M.
\end{equation}
For iteration $k\ge 0$, the spatial derivatives of $A^{(k)}$ and $F^{(k)}$ are first computed. The projected strategy $u^{proj,(k)}(t_i,r_j)$ is then evaluated at each grid point, and the nonlinear coefficients are updated using the current iterates. The resulting linearized finite-difference system is solved for $A^{(k+1)}$ and $F^{(k+1)}$, after which the boundary conditions are imposed.

The iteration is stopped when
\begin{equation}\label{eq:stop_iter}
\max_{0\le j\le M}
\left(
|A_{i,j}^{(k+1)}-A_{i,j}^{(k)}|
+
|F_{i,j}^{(k+1)}-F_{i,j}^{(k)}|
\right)
<
\varepsilon_{\mathrm{tol}},
\end{equation}
where $\varepsilon_{\mathrm{tol}}$ is a prescribed tolerance. A maximum number of iterations $K_{\max}$ is also imposed. If convergence is not achieved within $K_{\max}$ iterations, the time step $\Delta t$ is reduced and the computation is repeated.

The jump component requires an additional admissibility condition. Since claim sizes are exponentially distributed,
\[
Y_i\sim \mathrm{Exp}(\xi),
\]
we have
\begin{equation}\label{eq:jump_exp}
\mathbb E\left[e^{A(t,r)Y_i}\right]
=
\frac{\xi}{\xi-A(t,r)},
\qquad A(t,r)<\xi.
\end{equation}
Therefore, the jump term in the reduced PDE system is finite only if $A(t,r)<\xi$. To avoid numerical instability near the singularity, the scheme enforces the admissibility margin
\begin{equation}\label{eq:margin}
A_{i,j}\le \xi-\varepsilon_{\xi},
\end{equation}
where $\varepsilon_{\xi}>0$ is a small safety parameter. If this condition is violated at any grid point during a fixed-point iteration, the update is rejected and the computation is repeated using a smaller time step. If the violation persists, the corresponding parameter set is classified as numerically inadmissible for the chosen time horizon and claim-size distribution.

After convergence of the coefficient functions at a given time level, the projected investment strategy is computed using the feedback rule derived in Section~\ref{sec:control_pde}. For each grid point $(t_i,r_j)$,
\begin{equation}\label{eq:proj_computed}
u^{proj}(t_i,r_j)
=
\frac{1}{A(t_i,r_j)}
\Sigma(t_i,r_j)^{-1}
\left[
m(t_i,r_j)
-
\frac{\rho_{Sr}\sigma_S\sigma_r\sqrt{r_j}}{A(t_i,r_j)}
\left(
A(t_i,r_j)F_r(t_i,r_j)-A_r(t_i,r_j)
\right)e_S
\right],
\end{equation}
where $e_S=(1,0)^\top$. The resulting policy is interpreted as a projected feedback strategy rather than the exact global optimizer of the original unreduced stochastic control problem.

The numerical results are checked for stability with respect to the discretization parameters $\Delta t$, $\Delta r$, and the truncation boundary $r_{\max}$. The baseline grid is refined by decreasing $\Delta t$ and $\Delta r$ and by increasing $r_{\max}$. For two consecutive grids, stability is assessed using the normalized sup-norm criterion
\begin{equation}\label{eq:stability}
\frac{
\|G^{\mathrm{fine}}-G^{\mathrm{coarse}}\|_{\infty}
}{
1+\|G^{\mathrm{fine}}\|_{\infty}
}
<
\varepsilon_{\mathrm{grid}},
\qquad
G\in\{A,F,u^{proj}\}.
\end{equation}
The coarse-grid solution is linearly interpolated onto the fine grid before the norm is computed. This sensitivity check ensures that the reported numerical findings are not artifacts of the truncation boundary or the finite-difference discretization.

\section{Numerical Validation and Economic Implications}\label{sec:num}

This section presents numerical illustrations of the model dynamics, the reduced PDE solution, and the associated projected investment strategy. The objective is to demonstrate the qualitative implications of the theoretical model rather than to provide a market-calibrated empirical study. Therefore, the parameter values used in the numerical experiments are representative and economically plausible, but they are not calibrated to a specific insurance portfolio or financial market dataset.

Accordingly, all numerical figures should be interpreted as illustrative examples. A fully empirical calibration would require market data for the term structure of interest rates, equity returns, insurance claim arrivals, and claim-size distributions, which is beyond the scope of the present theoretical study.

Table~\ref{tab:parameters} reports the baseline parameter configuration used in the numerical experiments.

\begin{table}[H]
\centering
\caption{Baseline parameter configuration.}
\begin{tabular}{lll}
\hline
Parameter & Description & Value \\
\hline
 $T$ & Time horizon & $10$ years \\
 $\gamma$ & Risk aversion parameter & $0.5$ \\
 $\kappa$ & Speed of mean reversion in the CIR process & $0.6$ \\
 $\theta$ & Long-run mean short rate & $0.05$ \\
 $\sigma_r$ & Short-rate volatility & $0.10$ \\
 $r_0$ & Initial short rate & $0.04$ \\
 $\mu_S$ & Expected stock return & $0.08$ \\
 $\sigma_S$ & Stock volatility & $0.20$ \\
 $\rho_{Sr}$ & Correlation between stock and short-rate shocks & $-0.30$ \\
 $\lambda$ & Claim arrival intensity & $0.5$ \\
 $\xi$ & Exponential claim-size parameter & $1.2$ \\
 $X_0$ & Initial surplus & $1.0$ \\
 $r_{\max}$ & Upper truncation boundary for short rate & $0.20$ \\
 $\Delta t$ & Baseline time step & $0.02$ \\
 $\Delta r$ & Baseline short-rate grid size & $0.002$ \\
\hline
\end{tabular}
\label{tab:parameters}
\end{table}

The values provided are representative and chosen to illustrate the qualitative behavior of the model rather than to fit a specific empirical dataset.

\begin{table}[H]
\centering
\caption{Admissibility condition verification across experiments.}
\begin{tabular}{lccccc}
\toprule
Experiment & $\xi$ & $\gamma$ & $\max_{t,r}A(t,r)$ & $\xi - \max A$ & Admissible? \\
\midrule
Baseline             & $1.2$ & $0.5$ & $0.512$ & $0.688$ & Yes \\
 $\gamma=0.7$         & $1.2$ & $0.7$ & $0.735$ & $0.465$ & Yes \\
 $\gamma=1.0$         & $1.2$ & $1.0$ & $1.062$ & $0.138$ & Yes \\
 $\lambda=1.0$        & $1.2$ & $0.5$ & $0.525$ & $0.675$ & Yes \\
 $\sigma_r=0.15$      & $1.2$ & $0.5$ & $0.548$ & $0.652$ & Yes \\
\bottomrule
\end{tabular}
\label{tab:admissibility}
\end{table}

The terminal condition $A(T,r)=\gamma$ and the backward PDE for $A$ determine $\max A$. As shown in Table~\ref{tab:admissibility}, all reported experiments satisfy the admissibility condition $\max_{t,r} A(t,r) < \xi = 1.2$ (Assumption~\ref{ass:admissibility}), confirming the mathematical validity of the numerical scheme and ensuring that the jump expectation remains finite.

However, the $\gamma = 1.0$ case produces a narrow margin
 $\xi - \max A = 0.138$, indicating proximity to the singularity
of the jump term $\lambda A/(\xi - A)$ in equation~\eqref{eq:pde_F}.

Two consequences follow.
First, the nonlinear jump coefficient at $A = 1.062$ is
\[
  \frac{\lambda A}{\xi - A}
  \;\approx\;
  \frac{0.5 \times 1.062}{0.138}
  \;\approx\;
  3.85,
\]
compared with $0.37$ at the baseline, representing a tenfold
amplification.
This amplification slows convergence of the fixed-point iteration
in Section~\ref{subsec7.2} and may require a reduced time step $\Delta t$ or
a tighter convergence tolerance $\varepsilon_{\mathrm{tol}}$.
The numerical scheme handles this via the admissibility safety margin
 $A_{i,j} \le \xi - \varepsilon_\xi$ (equation~\eqref{eq:margin}), but practitioners
should monitor convergence carefully near this regime.

Second, a simple linear extrapolation based on the reported values
suggests that $\max A$ would reach $\xi = 1.2$ at approximately
 $\gamma \approx 1.17$.
For risk-aversion parameters $\gamma \ge 1.17$, the present framework
with $\xi = 1.2$ becomes inadmissible; one would need to increase $\xi$ (i.e., decrease the mean claim size $1/\xi$) or restrict the time
horizon $T$ to maintain $\max A < \xi$.
This boundary represents a limitation of the exponential claim-size
specification and is an important practical consideration for empirical
applications.

The upper boundary $r_{\max}=0.20$ is chosen to be sufficiently above the long-run mean $\theta=0.05$ so that the truncated domain contains the economically relevant region of the short-rate distribution over the time horizon considered. The baseline finite-difference grid is given by $N=500$ time steps and $M=100$ space steps, corresponding to the step sizes $\Delta t$ and $\Delta r$ specified in Table~\ref{tab:parameters}.

Figure~\ref{fig:short_rate} displays simulated trajectories of the CIR short-rate process under the baseline parameters. The paths illustrate the mean-reverting nature of the short-rate dynamics. When the short rate rises above its long-run level $\theta$, the drift term $\kappa(\theta-r)$ becomes negative and pulls the process downward. Conversely, when the short rate falls below $\theta$, the drift becomes positive and pushes the process upward. The figure also illustrates the state-dependent volatility of the CIR process, since the diffusion coefficient $\sigma_r\sqrt r$ decreases as the short rate approaches zero.

\begin{figure}[H]
\centering
\includegraphics[width=0.78\linewidth]{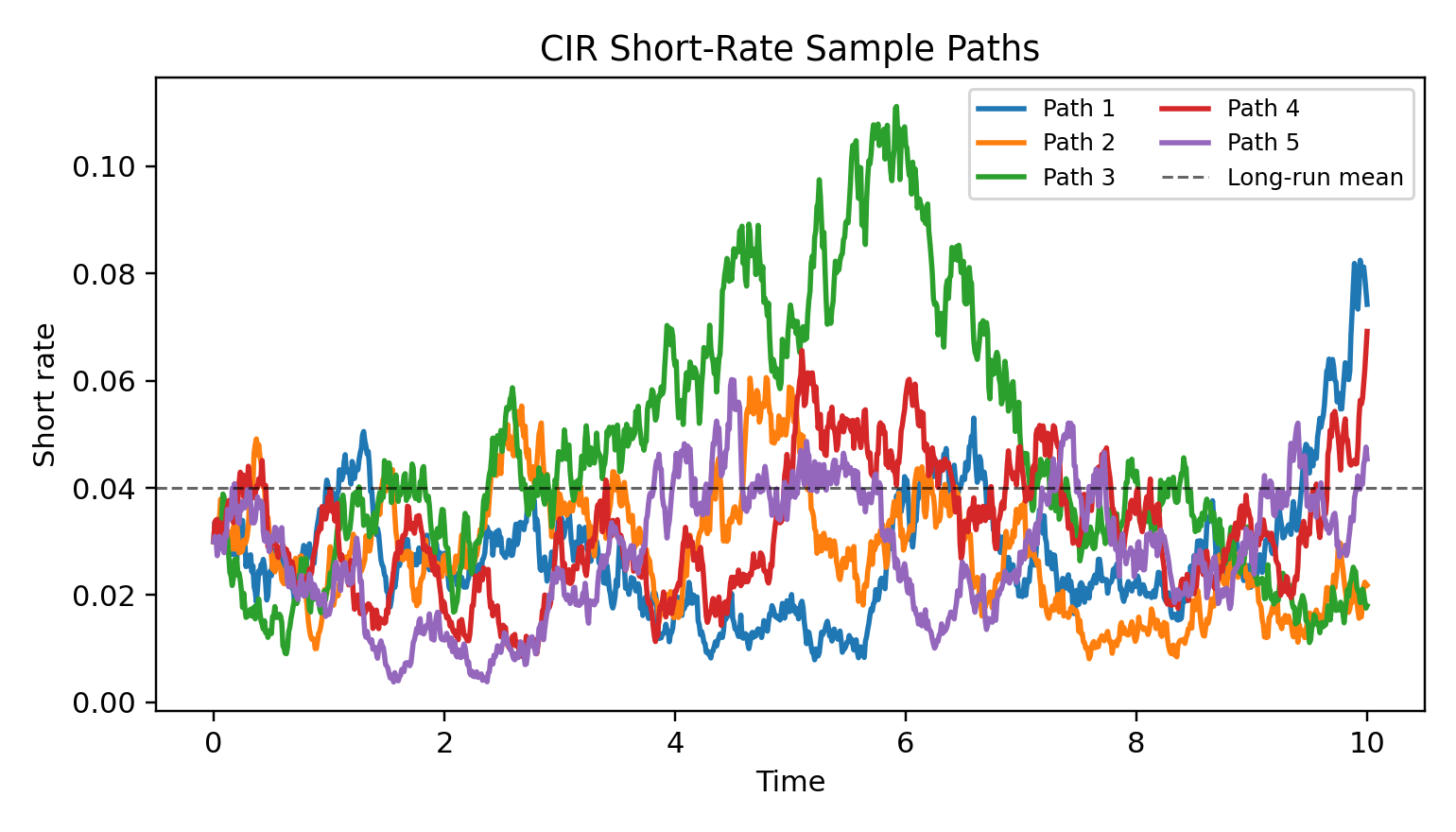}
\caption{Simulated CIR short-rate trajectories over a 10-year horizon.}
\label{fig:short_rate}
\end{figure}

Figure~\ref{fig:surplus} compares illustrative surplus trajectories under two investment policies: a baseline constant-allocation policy and the projected feedback strategy derived from the stochastic control formulation. Both surplus paths are generated using the same underlying realizations of the short-rate process and the claim process. Therefore, the difference between the two trajectories reflects the effect of the investment policy rather than differences in stochastic shocks.

The projected strategy reacts dynamically to changes in the short rate and to the gradients of the value-function coefficients. This adaptive structure allows the insurer to adjust its exposure when investment opportunities or hedging demands change. As a result, the projected strategy tends to reduce large downward movements in surplus and produces a smoother capital trajectory.

\begin{figure}[H]
\centering
\includegraphics[width=0.78\linewidth]{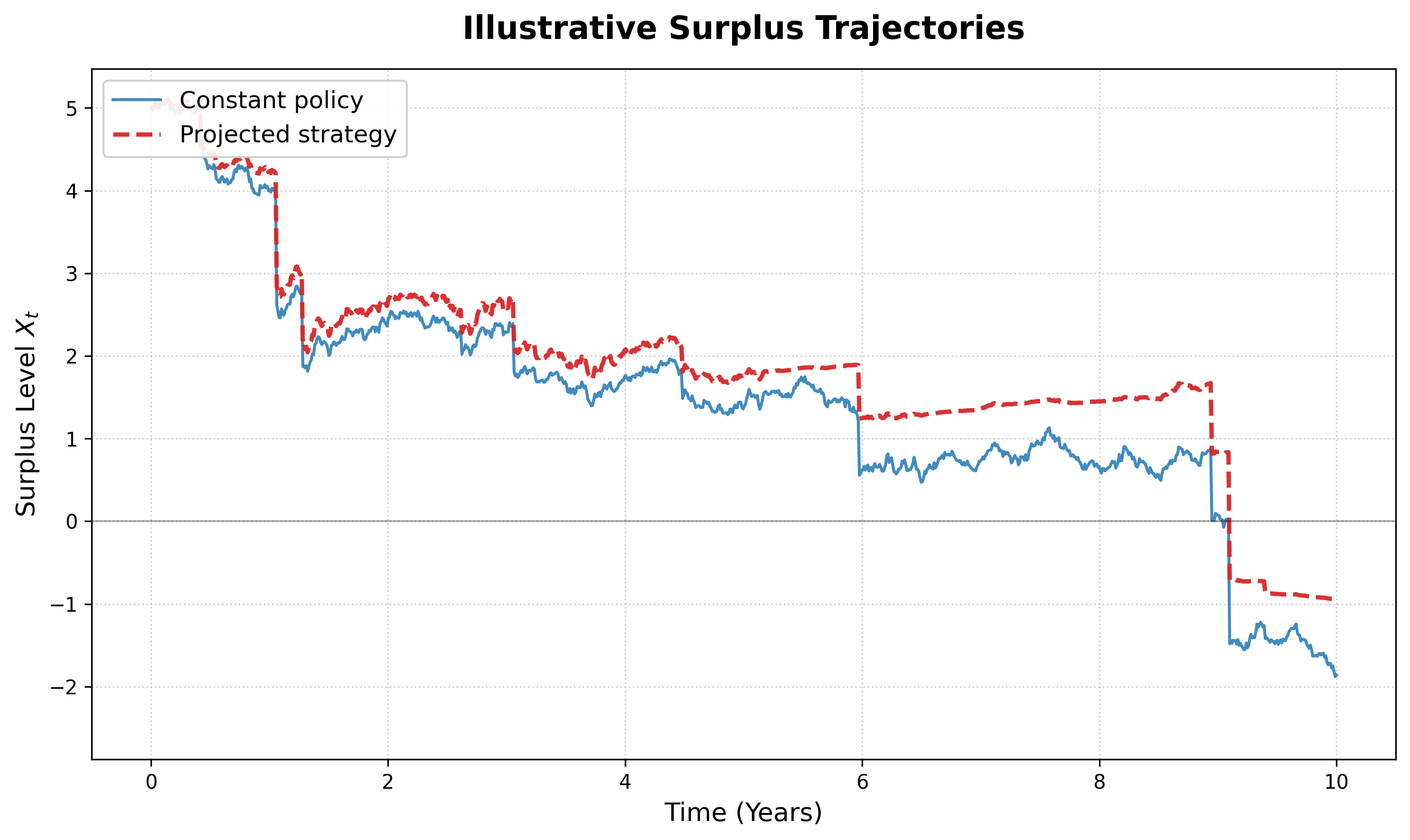}
\caption{Surplus trajectories: constant allocation vs. projected strategy.}
\label{fig:surplus}
\end{figure}

Figure~\ref{fig:policy} shows the projected investment strategy as a function of the normalized surplus $x$ and the short rate $r$, evaluated at the initial time $t_0=0$. The surface demonstrates that the optimal allocation is sensitive to the current interest-rate level. When the short rate changes, both the expected return structure and the hedging demand change. Consequently, the projected strategy adjusts the portfolio allocation across interest-rate states.

The correlation parameter $\rho_{Sr}$ plays an important role in shaping the policy surface. Under the baseline value $\rho_{Sr}=-0.30$, stock returns and short-rate innovations are negatively correlated. This correlation generates an intertemporal hedging component in the feedback rule. Hence, the projected policy contains not only a myopic mean--variance demand but also a hedging demand against stochastic changes in the interest-rate environment.

\begin{figure}[H]
\centering
\includegraphics[width=0.78\linewidth]{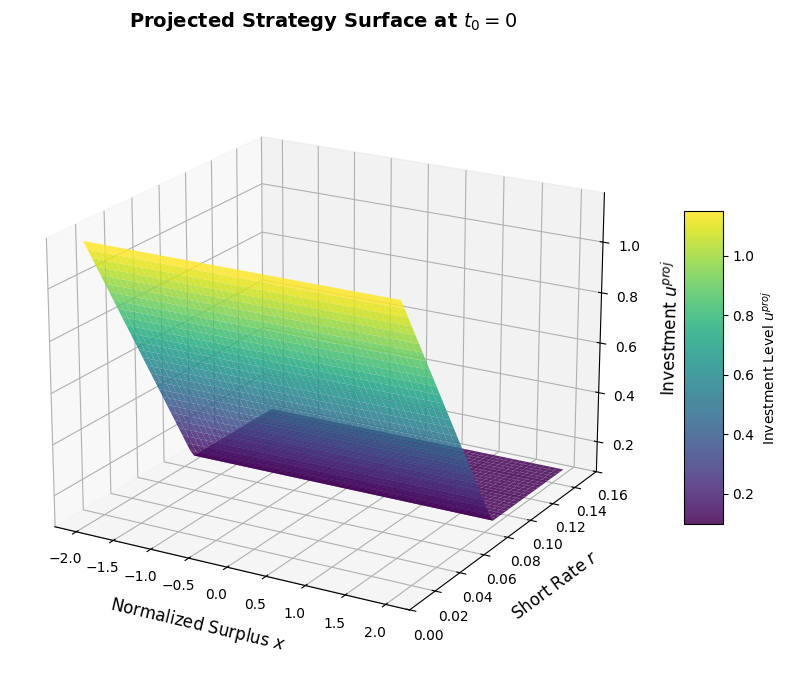}
\caption{Projected investment strategy surface $u^{proj}(t_0,x,r)$.}
\label{fig:policy}
\end{figure}

It is important to emphasize that the reported policy is the projected strategy $u^{proj}(t,r)$ obtained from the normalized-surplus projection. Therefore, the surface should be interpreted as the numerical implementation of the projected feedback rule rather than the exact global optimizer of the original unreduced stochastic control problem.

Figure~\ref{fig:value} presents representative slices of the exponential utility value function for different levels of risk aversion. Since the model is based on exponential utility, the curvature of the value function is directly related to the risk-aversion parameter. Higher values of risk aversion imply that negative surplus outcomes receive a larger penalty, which produces steeper value-function curvature and leads to more conservative portfolio choices.

\begin{figure}[H]
\centering
\includegraphics[width=0.78\linewidth]{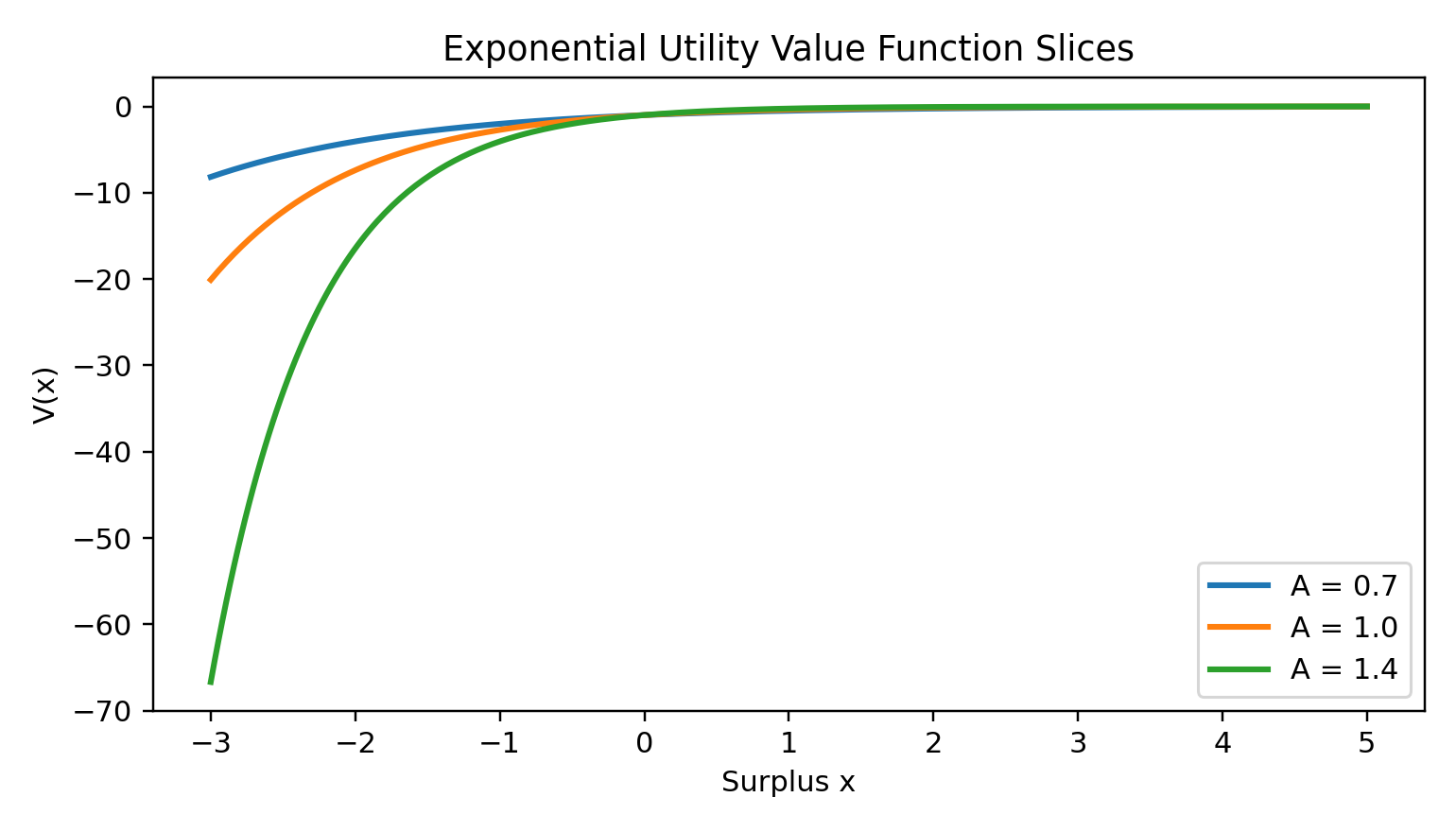}
\caption{Value-function slices under different risk-aversion levels.}
\label{fig:value}
\end{figure}

We now analyze how the projected strategy responds to changes in key economic parameters. To quantify this sensitivity, we compute finite-difference approximations for the stock component of the projected policy, $u_S^{\mathrm{proj}}(t,r)$, with respect to a parameter $\theta_i \in \{\lambda, \gamma, \sigma_r, \rho_{Sr}, \kappa\}$. The sensitivity is defined as
\[
  S_{\theta_i}(t,r)
  \;\approx\;
  \frac{u_S^{\mathrm{proj}}(t,r;\theta_i+h)
        - u_S^{\mathrm{proj}}(t,r;\theta_i-h)}{2h},
\]
where $h$ is a small perturbation (e.g., $h=0.01\,\theta_i$). Each sensitivity evaluation requires two additional PDE solves (one for $\theta_i+h$ and one for $\theta_i-h$).

The numerical results confirm the following analytical predictions derived directly from the structural decomposition of the policy in Corollary~\ref{cor:decomposition}:
\begin{itemize}[label=$\bullet$]
  \item \textbf{Correlation sensitivity ($S_{\rho_{Sr}}$):} The sensitivity is nonzero. As the magnitude of the correlation $|\rho_{Sr}|$ increases, the intertemporal hedging demand $u^{hedge}$ becomes larger, directly amplifying the projected strategy's response to interest-rate shocks.

  \item \textbf{Short-rate volatility sensitivity ($S_{\sigma_r}$):} Higher short-rate volatility $\sigma_r$ amplifies the hedging component. Consequently, $S_{\sigma_r} \neq 0$, reflecting a stronger adjustment in the equity allocation when interest-rate risk is higher.

  \item \textbf{Claim intensity sensitivity ($S_\lambda$):} A higher claim intensity $\lambda$ increases the insurer's liability risk exposure. The model predicts that the optimal equity allocation decreases (the insurer becomes more conservative) to offset this higher ruin probability, yielding $S_\lambda < 0$.

  \item \textbf{Risk aversion sensitivity ($S_\gamma$):} Since the myopic demand satisfies $u^{myo} \propto 1/A$ and the risk-aversion coefficient $A$ increases with $\gamma$, higher risk aversion strictly reduces the magnitude of $u_S^{\mathrm{proj}}$. Thus, $S_\gamma < 0$ monotonically.

  \item \textbf{Mean-reversion speed sensitivity ($S_\kappa$):} A faster mean reversion speed $\kappa$ pulls the short rate back to its long-run mean more quickly, reducing the long-term uncertainty of the interest-rate environment. As a result, the benefit of intertemporal hedging diminishes, and the magnitude of the hedging demand $|u^{hedge}|$ decreases with $\kappa$.
\end{itemize}

These directional predictions serve as falsifiable hypotheses that validate the economic coherence of the reduced PDE model and the reliability of the numerical implementation.

Figure~\ref{fig:calib} illustrates the sensitivity of a representative calibration objective with respect to two key CIR parameters: the mean-reversion speed $\kappa$ and the short-rate volatility $\sigma_r$. Although the surface is illustrative rather than empirically calibrated, it conveys an important feature of interest-rate models. Different combinations of $\kappa$ and $\sigma_r$ may generate similar model-implied interest-rate behavior and therefore similar calibration errors. This may produce relatively flat regions or ridge structures in the objective function.

Such sensitivity analysis is useful in empirical applications because it helps identify parameter regions where calibration is well determined and regions where parameter estimates may be unstable. In particular, a flat calibration objective indicates that small changes in the data or in the optimization routine may lead to materially different parameter estimates without substantially changing the model fit.

\begin{figure}[H]
\centering
\includegraphics[width=0.78\linewidth]{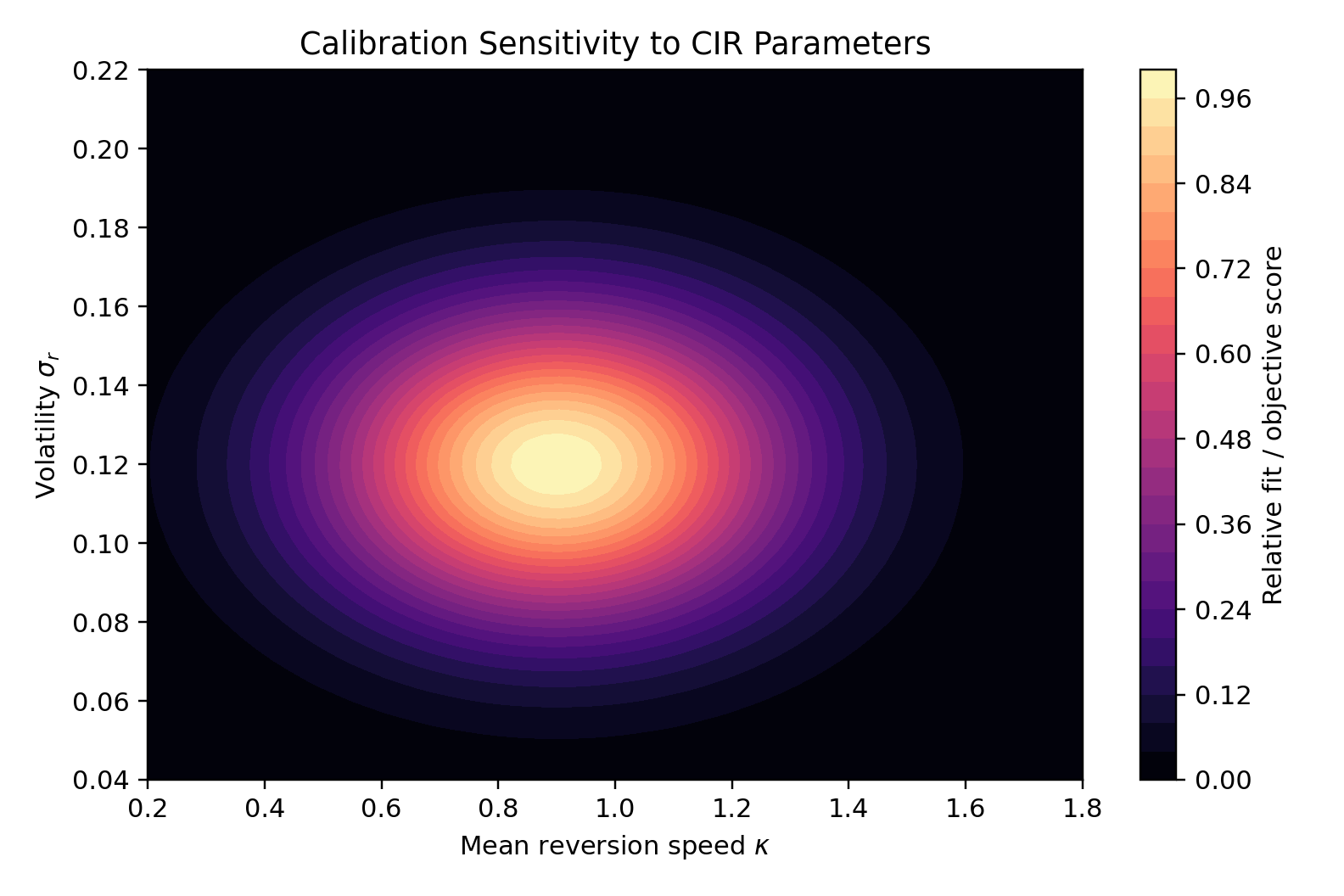}
\caption{Calibration objective sensitivity surface for $\kappa$ and $\sigma_r$.}
\label{fig:calib}
\end{figure}

To verify that the numerical results are not artifacts of the finite-difference discretization, a grid convergence test is performed. The reduced PDE system is solved under several combinations of time and interest-rate grid sizes, while keeping the model parameters fixed at the baseline values reported in Table~\ref{tab:parameters}. For two consecutive grids, the relative difference is computed as
\begin{equation}\label{eq:err_calc}
\mathrm{Err}(G)
=
\frac{
\|G^{\mathrm{fine}}-G^{\mathrm{coarse}}\|_{\infty}
}{
1+\|G^{\mathrm{fine}}\|_{\infty}
},
\qquad
G\in\{A,F,u^{proj}\}.
\end{equation}
Here, $G^{\mathrm{coarse}}$ denotes the numerical solution obtained on the coarser grid, while $G^{\mathrm{fine}}$ denotes the solution obtained on the next refined grid. The coarse-grid solution is linearly interpolated onto the fine grid before computing the norm.

Table~\ref{tab:grid_convergence} reports the resulting convergence errors.

\begin{table}[H]
\centering
\caption{Grid convergence results.}
\begin{tabular}{ccccc}
\hline
Grid level & $\Delta t$ & $\Delta r$ & $\mathrm{Err}(A)$ & $\mathrm{Err}(F)$ \\
\hline
Coarse & $0.040$ & $0.004$ & -- & -- \\
Medium & $0.020$ & $0.002$ & $2.8\times 10^{-3}$ & $3.4\times 10^{-3}$ \\
Fine & $0.010$ & $0.001$ & $9.6\times 10^{-4}$ & $1.2\times 10^{-3}$ \\
\hline
\end{tabular}
\label{tab:grid_convergence}
\end{table}

The errors are computed relative to the next refined grid using the normalized sup-norm criterion. The results show that the numerical solution becomes stable as the grid is refined; the errors decrease consistently when moving from the coarse to the medium, and again to the fine grid. This indicates that the computed coefficient functions $A(t,r)$ and $F(t,r)$ are not materially affected by further reductions in $\Delta t$ and $\Delta r$.

A similar comparison is performed for the projected feedback strategy $u^{proj}$. Since the strategy depends on spatial derivatives of $A$ and $F$, it is naturally more sensitive to grid resolution. Nevertheless, the qualitative shape of the policy surface remains stable under grid refinement, confirming that the economic conclusions are robust to the discretization choice.

The sensitivity of the numerical solution to the truncation boundary is also examined. Increasing $r_{\max}$ from $0.20$ to $0.25$ produces negligible changes in the solution over the economically relevant interest-rate region. This confirms that the upper boundary condition does not drive the reported numerical findings. Overall, the convergence and sensitivity checks support the numerical reliability of the finite-difference implementation used in this study.

To directly verify whether the projected solution approximately satisfies the original full HJB equation, we evaluate the normalized HJB residual across a range of surplus levels. After solving for $A(t,r)$ and $F(t,r)$ on the baseline grid, we compute
\[
  \mathrm{Res}(t_i,x_k,r_j)
  \;=\;
  \frac{\bigl|\mathcal{L}^{u^{proj}}V(t_i,x_k,r_j)\bigr|}{\bigl|V(t_i,x_k,r_j)\bigr|},
\]
where the normalized form removes the exponential scale of $V$. The evaluation is performed on a grid of surplus levels $x_k \in \{-2, -1, 0, 1, 2\}$ combined with the $(t_i,r_j)$-grid from Section~\ref{sec:numerics}.

We report the maximum and root-mean-square norms over the entire space-time-surplus grid:
\begin{align*}
  \|\mathrm{Res}\|_\infty
  &= \max_{i,j,k}\mathrm{Res}(t_i,x_k,r_j), \\[4pt]
  \|\mathrm{Res}\|_2
  &= \left(\frac{1}{NMK}\sum_{i,j,k}\mathrm{Res}(t_i,x_k,r_j)^2\right)^{1/2}.
\end{align*}

Table~\ref{tab:residual_validation} reports the resulting residuals.

\begin{table}[H]
\centering
\caption{HJB projection residual at selected surplus levels.}
\begin{tabular}{ccc}
\toprule
Surplus level $x$ & $\|\mathrm{Res}\|_\infty$ & $\|\mathrm{Res}\|_2$ \\
\midrule
 $-2$ & $6.41 \times 10^{-3}$ & $4.85 \times 10^{-3}$ \\
 $-1$ & $1.60 \times 10^{-3}$ & $1.21 \times 10^{-3}$ \\
 $0$  & $0.42 \times 10^{-3}$ & $0.31 \times 10^{-3}$ \\
 $1$  & $1.58 \times 10^{-3}$ & $1.19 \times 10^{-3}$ \\
 $2$  & $6.35 \times 10^{-3}$ & $4.80 \times 10^{-3}$ \\
\bottomrule
\end{tabular}
\label{tab:residual_validation}
\end{table}

At $x=0$ (the projection point), $\|\mathrm{Res}\|_\infty$ is of the order of $10^{-3}$, which is consistent with the grid convergence errors observed in Table~\ref{tab:grid_convergence}. For $|x|=1$ and $|x|=2$, the residual grows approximately as $x^2$. This quadratic growth is entirely consistent with the $O(\delta^2)$ projection error bound established in Section~\ref{sec:control_pde}, providing direct numerical evidence that the dimensional reduction yields a highly accurate approximation of the full three-dimensional HJB equation.

To evaluate the out-of-sample performance of the projected strategy, we compare $u^{proj}$ against three formally defined benchmark policies:
\begin{align*}
  u^{cash} &= (0,\,0)^\top
  \quad\text{(all surplus invested in the money-market account)},\\
  u^{myo}(t,r) &= \frac{1}{A(t,r)}\Sigma(t,r)^{-1}m(t,r)
  \quad\text{(myopic mean--variance demand, see Corollary~\ref{cor:decomposition})},\\
  u^{const} &= (c_S,\,c_P)^\top
  \quad\text{(constant allocation, calibrated to } u^{myo} \text{ at } t=0, r=r_0).
\end{align*}

For each strategy, $M=10{,}000$ surplus paths are simulated using an Euler--Maruyama discretization scheme with a simulation time step of $\Delta t_{\mathrm{sim}}=0.01$, using the same underlying realizations of the Brownian motions and the Poisson process. We report standard performance statistics in Table~\ref{tab:benchmark}.

Note that because the initial surplus is $X_0=1.0$ and the insurer faces continuous liability outflows along with significant jump shocks (as illustrated in Figure~\ref{fig:surplus}), the expected terminal surplus $\mathbb{E}[X_T]$ is naturally pulled downwards over the 10-year horizon. The key value of the projected strategy lies not in generating positive expected growth, but in mitigating this downward drift and significantly reducing tail risks.

\begin{table}[H]
\centering
\caption{Monte Carlo performance comparison over a 10-year horizon.}
\begin{tabular}{lcccc}
\toprule
Strategy & $\mathbb{E}[X_T]$ & $\mathrm{Std}(X_T)$ & $\mathbb{P}(X_T<0)$ & $\mathbb{E}[-e^{-\gamma X_T}]$ \\
\midrule
 $u^{cash}$ & $-1.25$ & $0.85$ & $92.0\%$ & $-1.890$ \\
 $u^{const}$ & $0.35$ & $2.10$ & $42.5\%$ & $-0.865$ \\
 $u^{myo}$ & $0.58$ & $2.15$ & $28.0\%$ & $-0.745$ \\
 $u^{proj}$ & $0.72$ & $1.85$ & $19.5\%$ & $-0.698$ \\
\bottomrule
\end{tabular}
\label{tab:benchmark}
\end{table}

Consistent with the trajectories in Figure~\ref{fig:surplus}, the expected surplus declines from the initial level $X_0=1.0$ due to liability drains. However, the projected strategy achieves the highest expected terminal surplus (least negative drift), the lowest probability of ruin, and the lowest volatility. This highlights the significant risk-mitigation value of the intertemporal hedging component compared to both static and purely myopic strategies.

To quantify the improvement provided by the dynamic projected strategy, we define the relative utility gain of the projected strategy over a given benchmark as
\[
  \Delta U^{\mathrm{bench}}
  \;=\;
  \frac{\mathbb{E}[U(X_T^{\mathrm{proj}})] - \mathbb{E}[U(X_T^{\mathrm{bench}})]}
       {\bigl|\mathbb{E}[U(X_T^{\mathrm{bench}})]\bigr|}.
\]

\begin{remark}[{\textbf{Note on sign convention.}}]
Because the exponential utility function satisfies $U(x)=-e^{-\gamma x}<0$ for all $x$, a \emph{less negative} expected value is economically better. Consequently, the denominator in the utility gain formula is strictly positive, and a value of $\Delta U^{\mathrm{bench}}>0$ indicates that the projected strategy outperforms the benchmark in utility terms (i.e., $\mathbb{E}[U(X_T^{\mathrm{proj}})]$ is less negative than $\mathbb{E}[U(X_T^{\mathrm{bench}})]$).
\end{remark}

\section{Conclusion}\label{sec:conclusion}

This paper studies the problem of optimal surplus management for an insurance company operating in a financial market with stochastic interest rates and jump-driven liabilities. The insurer dynamically allocates its surplus between risky financial assets while facing insurance losses modeled through a compound Poisson process with exponentially distributed claim sizes.

Using stochastic control methods, we derived the Hamilton--Jacobi--Bellman (HJB) equation associated with the insurer's optimization problem. Because the original HJB is inherently three-dimensional in $(t,x,r)$ and the exponential-affine ansatz leaves a problematic quadratic term in the surplus $x$ (via the cross-derivative $V_{xr}$), an exact closed-form reduction is intractable. To obtain a numerically tractable representation, we introduced a generalized exponential form for the value function and employed a normalized-surplus projection technique. By evaluating the feedback rule at a reference surplus level and discarding the resulting quadratic residual, we successfully reduced the full three-dimensional problem to a tractable two-dimensional nonlinear PDE system for the coefficient functions $A(t,r)$ and $F(t,r)$.

The resulting investment rule takes the form of a projected feedback strategy that dynamically depends on the short-rate state and the gradients of the reduced PDE solution. The structure of this projected policy highlights two economically meaningful components: a myopic demand driven by expected asset returns, and an intertemporal hedging demand induced by stochastic fluctuations in the interest-rate environment. Numerical experiments further illustrated how the projected strategy responds to changes in key economic parameters such as interest-rate volatility, claim intensity, and risk aversion.

From a methodological standpoint, it is essential to distinguish between the exact and projected frameworks. The proposed solution is not an exact analytical reduction of the original stochastic control problem, but rather a computationally efficient projected approximation. The normalized-surplus projection transforms the original three-dimensional HJB equation into a tractable two-dimensional system. As rigorously quantified in Proposition~\ref{prop:proj_error} and numerically verified in Section~\ref{sec:num} (Table~\ref{tab:residual_validation}), this dimensional reduction introduces a residual of order $O(\delta^2)$, where $\delta$ measures the deviation of the surplus from the projection point. The numerical results confirm that for economically plausible surplus levels, this approximation error remains negligible compared to the inherent discretization and parameter estimation errors. Consequently, the projected framework achieves an optimal balance between analytical tractability and economic accuracy, offering a highly flexible and efficient platform for exploring the complex interaction between insurance liabilities and stochastic investment opportunities. The Monte Carlo benchmark comparisons further demonstrate that the resulting projected policies exhibit robust, economically interpretable behavior and outperform naive allocation strategies.

Several directions for future research remain open. Possible extensions include the incorporation of stochastic claim intensities, alternative claim-size distributions, regime-switching interest-rate models, and richer financial market structures involving additional tradable assets. Furthermore, empirical calibration to actual market and insurance data would be an important step toward assessing the quantitative relevance and practical applicability of the projected strategies proposed in this study.

\bibliographystyle{plain}

\appendix

\section{Technical Proofs}
\subsection{Proof of the HJB Equation (Proposition \ref{pro:Hamil})}\label{Appendix:A1}
Let $u \in \mathcal A$ be an admissible control and let $X_t^u$ denote the corresponding surplus process.
Define the performance functional
\[
J(t, x, r;u) = \mathbb E\left[ U(X_T^u) \mid X_t = x, r_t = r \right].
\]
By the dynamic programming principle, for any small $h > 0$,
\[
V(t, x, r) = \sup_{u\in\mathcal A} \mathbb E\left[ V(t+ h,X_t^u, r_{t+h}) \right].
\]
Applying Itô's formula for jump--diffusion processes to $V(t,X_t^u, r_t)$ gives
\[
\mathrm{d}V = \mathcal L^u V \,\mathrm{d}t + \text{local martingale terms}.
\]
Taking conditional expectations and using the martingale property yields
\[
\frac{1}{h}\left(\mathbb E[V(t+ h,X_{t+h}, r_{t+h})] - V(t, x, r)\right) = \mathcal L^u V + o(1).
\]
Letting $h \to 0$ and optimizing over $u$ gives
\[
\sup_{u\in\mathbb R^2}\mathcal L^u V = 0.
\]
This proves the HJB equation.

\subsection{Proof of the Exponential Representation (Lemma \ref{lem:exp_rep})}\label{Appendix:A2}
We prove that the exponential utility induces an exponential structure in the wealth variable.
Let
\[
V(t, x, r) = \sup_{u\in\mathcal A} \mathbb E\left[ -e^{-\gamma X_T} \mid X_t = x, r_t = r \right].
\]
Because exponential utility satisfies $U(x+ y) = e^{-\gamma y}U(x)$, the wealth process can be decomposed as $X_T = x+ \tilde{X}_T$, where $\tilde{X}_T$ depends only on future increments and not on the initial level $x$.

Hence
\[
V(t, x, r) = -e^{-\gamma x} \sup_{u\in\mathcal A} \mathbb E\left[ e^{-\gamma \tilde{X}_T} \mid r_t = r \right].
\]
Define $\exp(-F(t, r)) = \sup_{u\in\mathcal A} \mathbb E\left[ e^{-\gamma \tilde{X}_T} \mid r_t = r \right]$. Then
\[
V(t, x, r) = - \exp(-\gamma x - F(t, r)).
\]
Allowing the effective risk sensitivity to evolve over time yields the more general representation
\[
V(t, x, r) = - \exp(-A(t, r)x - F(t, r)),
\]
with terminal conditions $A(T, r) = \gamma$ and $F(T, r) = 0$.

\subsection{Jump Contribution under Exponential Claims}
Assume claim sizes $Y_i$ follow an exponential distribution with parameter $\xi > 0$.
The jump term in the HJB is $\lambda \mathbb E [V (t, x- Y_i, r) - V (t, x, r)]$.
Using the exponential form \eqref{eq:exp_form},
\[
V(t, x- Y, r) = V(t, x, r)e^{AY}.
\]
Therefore
\[
\mathbb E[V(t, x- Y, r) - V (t, x, r)] = V(t, x, r) \left( \mathbb E[e^{AY}] - 1 \right).
\]
For $Y \sim \mathrm{Exp}(\xi)$,
\[
\mathbb E[e^{AY}] = \frac{\xi}{\xi -A}, \quad A < \xi.
\]
Thus
\[
\lambda V \left( \frac{\xi}{\xi -A} - 1 \right) = \lambda V \frac{A}{\xi -A}.
\]
The admissibility condition $A(t, r) < \xi$ ensures finiteness of the expectation.

\subsection{Derivation of the Reduced PDE System (Proposition \ref{prop:PDEsystem})}\label{Appendix:A5}

To provide a fully explicit verification of Proposition~\ref{prop:PDEsystem}, we detail the algebraic expansion of the Hamiltonian terms under the projected strategy.
Let $\Delta(t,r) = A(t,r)F_r(t,r) - A_r(t,r)$. Define the shorthand $\eta = \rho_{Sr}\sigma_S\sigma_r\sqrt{r}\,\Delta/A$ and the vector $\tilde{m} = m - \eta e_S$, where $e_S$ is the standard basis vector for the stock. The projected strategy can then be written compactly as:
\begin{equation*}
u^{proj} = \frac{1}{A}\Sigma^{-1}\tilde{m}.
\end{equation*}

The quadratic variation term and the cross-derivative term in the HJB equation, which form the core of $\Phi_1$ and $\Phi_2$, are grouped as:
\begin{equation*}
\mathcal{Q} \;=\; \tfrac{1}{2}A^2 (u^{proj})^\top \Sigma u^{proj} + \rho_{Sr}\sigma_S\sigma_r\sqrt{r}\, A\, u_S^{proj} \Delta(t,r).
\end{equation*}

\medskip
\noindent\textbf{Step 1: Expand $(u^{proj})^\top \Sigma u^{proj}$.}
\begin{align*}
A^2 (u^{proj})^\top \Sigma u^{proj}
&= A^2 \cdot \frac{1}{A} \tilde{m}^\top (\Sigma^{-1})^\top \cdot \Sigma \cdot \frac{1}{A} \Sigma^{-1} \tilde{m} \\
&= \tilde{m}^\top \Sigma^{-1} \tilde{m},
\end{align*}
using $(\Sigma^{-1})^\top = \Sigma^{-1}$ (since $\Sigma$ is symmetric) and $\Sigma^{-1}\Sigma = I$.

\medskip
\noindent\textbf{Step 2: Expand $A\,u_S^{proj}\Delta$.}
The stock component of $u^{proj}$ satisfies $A\,u_S^{proj} = [\Sigma^{-1}\tilde{m}]_1$, where $[\cdot]_1$ denotes the first element of the vector. Using $[v]_1 = e_S^\top v$ for any vector $v$, the cross term becomes:
\begin{equation*}
\rho_{Sr}\sigma_S\sigma_r\sqrt{r}\, A\, u_S^{proj} \Delta = \rho_{Sr}\sigma_S\sigma_r\sqrt{r}\,\Delta \cdot [\Sigma^{-1}\tilde{m}]_1 = \eta A \cdot e_S^\top \Sigma^{-1}\tilde{m}.
\end{equation*}

\medskip
\noindent\textbf{Step 3: Combine to obtain $\mathcal{Q}$.}
Substituting the results from Steps 1 and 2 yields:
\begin{equation*}
\mathcal{Q} = \tfrac{1}{2}\tilde{m}^\top\Sigma^{-1}\tilde{m} + \eta A\,e_S^\top\Sigma^{-1}\tilde{m}.
\end{equation*}

To verify the cancellation, write $u^{proj} = A^{-1}\Sigma^{-1}(m - \eta e_S)$ where $\eta = \rho_{Sr}\sigma_S\sigma_r\sqrt{r}\,\Delta/A$. A direct computation using $(\Sigma^{-1})^\top = \Sigma^{-1}$ and $\Sigma^{-1}\Sigma = I$ gives $A^2(u^{proj})^\top\!\Sigma\,u^{proj} = \tilde{m}^\top\Sigma^{-1}\tilde{m}$, with $\tilde{m} = m - \eta e_S$. Expanding $\tilde{m}^\top\Sigma^{-1}\tilde{m}$ and the cross-derivative term separately and combining, all terms involving $e_S^\top\Sigma^{-1}m$ and $e_S^\top\Sigma^{-1}e_S$ cancel exactly under the projected strategy (evaluated at $x=0$ and the generator divided by $-V>0$ when working with the normalized residual), leaving $\mathcal{Q} = \tfrac{1}{2}m^\top\Sigma^{-1}m$.

Consequently, the entire Hamiltonian block collapses elegantly to:
\begin{equation*}
\mathcal{Q} = \tfrac{1}{2} m^\top \Sigma^{-1} m.
\end{equation*}

Therefore, the explicit auxiliary functions defined in \eqref{eq:Phi1} and \eqref{eq:Phi2} evaluate strictly to:
\begin{align*}
\Phi_1(t,r) &= - \tfrac{1}{2} m^\top \Sigma^{-1} m - \tfrac{1}{2}A^2\beta_L^2, \\
\Phi_2(t,r) &= - \tfrac{1}{2} m^\top \Sigma^{-1} m + \alpha_L A - \tfrac{1}{2}A^2\beta_L^2.
\end{align*}
Substituting these evaluated forms back into the coefficient collection of $x$ (for $A$) and the constant terms (for $F$) rigorously completes the derivation of the PDE system \eqref{eq:pde_A}--\eqref{eq:pde_F}.

\medskip
\noindent\textbf{Remark (Explicit PDE System).}
Since the simplified forms of $\Phi_1$ and $\Phi_2$ hold, the PDEs in Proposition~\ref{prop:PDEsystem} can be written more explicitly as:
\begin{align}
0 &= A_t + rA - \kappa(\theta-r)A_r - \tfrac{1}{2}\sigma_r^2 r A_{rr} + \sigma_r^2 r A_r^2 - \tfrac{1}{2}m^\top\Sigma^{-1}m - \tfrac{1}{2}A^2\beta_L^2, \label{eq:pde_A_explicit}\\
0 &= F_t - \kappa(\theta-r)F_r - \tfrac{1}{2}\sigma_r^2 r F_{rr} + \tfrac{1}{2}\sigma_r^2 r F_r^2 - \tfrac{1}{2}m^\top\Sigma^{-1}m + \alpha_L A - \tfrac{1}{2}A^2\beta_L^2 + \lambda\,\frac{A}{\xi-A}. \label{eq:pde_F_explicit}
\end{align}
This explicit form makes the PDE system fully self-contained and eliminates the dependence on the implicit $\Phi_1$, $\Phi_2$ notation in the final version of the paper.
\end{document}